\documentclass{article}
\usepackage{ijcai21}

\usepackage{times}
\usepackage{soul}
\usepackage{url}
\usepackage[hidelinks]{hyperref}
\usepackage[utf8]{inputenc}
\usepackage[small]{caption}
\usepackage{graphicx}
\usepackage{amsmath}
\usepackage{amsthm}
\usepackage{booktabs}
\usepackage{algorithm}
\usepackage{algorithmic}
\usepackage{enumitem}
\usepackage{amsfonts}
\usepackage{bbm}
\usepackage{bm}
\usepackage{subfig}
\usepackage{multirow}
\usepackage{amssymb}
\usepackage[title]{appendix}
\newtheorem{theorem}{Theorem}
\newtheorem{lemma}{Lemma}
\newtheorem{corollary}{Corollary}

\newtheorem{dfn}{Definition}

\title{Mean Field Equilibrium in Multi-Armed Bandit Game with Continuous Reward} 
\author{
Xiong Wang$^1$
\and
Riheng Jia$^2$
\affiliations
$^1$The Chinese University of Hong Kong, Hong Kong SAR, China\\
$^2$Zhejiang Normal University, Jinhua, China\\
\emails
xwang@cse.cuhk.edu.hk,
rihengjia@zjnu.edu.cn
\thanks{Corresponding author: Riheng Jia}
}

\renewcommand\footnotemark{}
\begin{document}

\maketitle

\begin{abstract}
Mean field game facilitates analyzing multi-armed bandit (MAB) for \emph{a large number of agents} by approximating  their interactions with  an average effect. Existing mean field models for multi-agent  MAB  mostly assume a binary reward function, which leads to tractable analysis but is usually not applicable in practical scenarios. In this paper, we study the mean field bandit game with a \emph{continuous reward function}. Specifically, we focus on deriving the existence and uniqueness of \emph{mean field equilibrium} (MFE), thereby guaranteeing the asymptotic stability of the multi-agent system. To accommodate the continuous reward function, we encode the learned reward into an agent state, which is in turn mapped to its stochastic arm playing policy and updated using realized observations. We show that the state evolution is upper semi-continuous, based on which the existence of MFE is obtained. As the Markov analysis is mainly for the case of discrete state, we transform the stochastic continuous state evolution into a \emph{deterministic} ordinary differential equation (ODE). On this basis, we can characterize a contraction mapping for the ODE to ensure a unique MFE for the bandit game. Extensive evaluations validate our MFE characterization, and exhibit tight empirical regret of the MAB problem.
\end{abstract}

\section{Introduction}

Great efforts have been devoted to multi-armed bandit (MAB) for sequential decision making, where agents can only observe limited information when pulling arms~\cite{Ref:regret2012bubeck}. Though principled, these techniques are mostly suitable for the single-agent scenario. As multiple agents may coexist, multi-agent learning is proposed to investigate agent interactions ~\cite{Ref:cooperative2005panait}, and particularly Markov game serves as the main tool to characterize the learning equilibrium in multi-agent systems ~\cite{Ref:markov1994littman}. Perfect information is usually needed in Markov game to determine the learning strategy, which makes it inefficient to analyze the equilibrium when the number of agents scales or when dealing with MAB problems. To accommodate this bandit feedback,  Hart and Mas-Colell propose a no-regret learning method~\cite{Ref:a2000hart}, yet only the coarse correlated  equilibrium can be derived for a handful of agents.

Mean field game is an effective model to approximate complex interactions among large populations~\cite{Ref:mean2007lasry,Ref:large2006huang}, where the approximation error is $O(1/\sqrt{N})$ with $N$ the number of agents~\cite{Ref:on2016ying}. Combined with online learning framework, mean field model is applied to multi-agent systems for deciding agents' strategic actions  and characterizing system's stable state~\cite{Ref:learning2017Cardaliaguet,iyer2014mean}. Existing mean field analysis often requires  complete knowledge of not only the reward function, but also the historical information to obtain the mean field equilibrium (MFE), which is not applicable in scenarios with merely limited feedback available. 

The problem of analyzing bandit feedback for many agents remains open until Gummadi \emph{et al.} use the mean field model to study multi-agent MAB in a repeated game~\cite{Ref:mean2013gummadi}, where they derive a unique MFE based on assumptions of binary reward and state regeneration. Following works have adopted their model in cellular network~\cite{Ref:distributed2017maghsudi} and smart grid~\cite{Ref:intelligent2017zhao}. However, the binary reward setting is too restricted for real-world multi-agent systems. Like in a resource competition game, agents usually share the resource with each other instead of occupying exclusively, thus the reward is a continuous value in $[0,1]$ rather than only $0$ or $1$~\cite{yang2018mean,hanif2015mean}. Besides,  agent state is assumed to regenerate with a certain probability to deduce an equilibrium in these works, whereas typical repeated games mostly entail iterative plays of each agent with no  regeneration~\cite{Ref:cooperative2005panait}. 
One critical question then is how to achieve the \emph{equilibrium for the non-regenerated bandit game involving a large number of agents with a generalized continuous reward}.

In this paper, we propose a mean field model to tackle the bandit game in  large-population multi-agent systems with a continuous reward function. We aim  to resolve  the \emph{existence and uniqueness} issues of  MFE, and hence are faced with the following challenges. First, characterizing the learning equilibrium is inconsistent with minimizing the regret since a non-regenerated bandit game is unstable under classical MAB algorithms. Therefore, one needs to model the agent state  cautiously to  strike a balance between system stability and tight-bounded regret. Second, existing Markov game is ineffective to track the state evolution, which is identified in a continuous regime instead of a discrete value due to the generalized reward function. Third,  existence and uniqueness of MFE is hard to derive as we only observe a bandit feedback.

To handle these challenges, we encode and update an agent state using observed rewards, and devise a stationary policy to  map the state to stochastic arm playing strategy  in turn. We show that the state evolution satisfies a \emph{fixed point} condition, thus proving the convergence of states to an existing MFE. Then, we develop a stochastic approximation to transform the stochastic state evolution into a \emph{deterministic ordinary differential equation} (ODE), so that we can derive  a contraction mapping for the ODE to obtain the unique MFE. Finally, we deduce a regret related cumulative state change, and extend the mean field model. Our contributions are summarized:

\begin{itemize}[leftmargin=*]
	\item We propose a new framework of mean field analysis to explore the multi-agent bandit game with a \emph{continuous reward}. Our framework can \emph{generalize} the previous binary reward  function  to a more universal scenario.
	\item We characterize the stable state for a large-population stochastic system with only limited information by deriving the existence and uniqueness of MFE. Specifically, we show there exists a MFE via \emph{Kakutani fixed-point} theorem, and further devise both rigorous and relaxed conditions for $||\cdot||_{\infty}$-\emph{contraction} to obtain the unique MFE, thereby providing a tractable and guaranteed system performance.
	\item We encode the learned reward into agent state, which can both ensure the \emph{system stability} and yield a \emph{tight empirical regret}. Model extensions also reveal the \emph{robustness} of our mean field analysis in different variants of MAB problems.
\end{itemize}
\vspace{-5pt}
\section{Model and Setup} \label{Sec: SysMdl}
We study the repeated bandit game in a multi-agent system, where there are a large number of agents $\mathcal{N}=\{1,2,...,N\}$ and  a set of arms (actions)  $\mathcal{M} = \{1,2,...,M\}$. Time is discretized into slots of equal length $\{0,1,...,n,...\}$. 

\subsection{State and Playing Policy}
The state of each agent $i$ encodes its \emph{learned reward} based on the observations realized so far, which is denoted as $s^i_n = [s^i_n(1),s^i_n(2),...,s^i_n(M)] \in \mathbb{R}^M$ with $s^i_n(j)$ being the learned reward of arm $j$ upon to time slot $n$. 

Agents are assumed to follow a \emph{stationary policy} when solving MAB problems in the repeated bandit game~\cite{Ref:learning2017heliou}. Denote a  simplex $\Delta^{M-1}=\{z \in [0,1]^M: \sum_{j=1}^{M}z(j) =1\}$ as  the probability distribution over $M$ arms. An arm playing policy is a mapping from the state space to the simplex $\sigma: \mathbb{R}^M \rightarrow  \Delta^{M-1}$, i.e., $\sigma (s^i_n,j)$ means the probability that agent $i$ would choose arm $j$ with $\sum_{j=1}^M \sigma(s^i_n,j) =1$. In particular, we focus on a Hedge stationary policy below:
\begin{equation} \label{Eq:hedge}
\sigma(s^i_n,j) = (1-\eta) \frac{\mathrm{Exp}\left(\beta s^i_n(j)\right)}{\sum_{k=1}^M \mathrm{Exp} \left(\beta s^i_n(k)\right)} + \frac{\eta}{M},
\end{equation}
where $\mathrm{Exp}(\cdot)$ represents the exponential function and $\beta >0$ is the smoothing parameter. The value $\sigma(s^i_n,j)$ consists of two parts balanced by parameter $\eta \in [0,1]$, in which the first part is a logit policy and the second is a random selection.

\subsection{Mean Field Reward}
Let $a^i_n$ be  the played arm of agent $i$ following  policy $\sigma (s^i_n)$.  Population profile $f_n=[f_n(1),f_n(2),...,f_n(M)]$ indicates the proportion of agents playing the
various arms in time slot $n$. Hence, the $j$-th element $f_n(j)$ is defined:
\vspace{-2pt}
\begin{equation} \label{Eq:Population}
f_n(j)= \frac{1}{N} \sum \nolimits_{i=1}^{N} \mathbbm{1}_{\{a^i_n = j\}},
\vspace{-3pt}
\end{equation}
where $\mathbbm{1}_{\{\cdot\}}$ is the indicator function.

Each agent's reward of playing an arm is determined by the actions of all agents due to their interactions.  For instance,  in a resource competition game, the reward will decrease if more competitors (agents) simultaneously compete for the same resource (arm). As directly characterizing the influence of agents' actions is difficult for large populations owing to the curse of dimensionality, we employ the mean field model to approximate interactions among agents, and accordingly their rewards will depend on the population profile $f_n$~\cite{Ref:mean2018yang,Ref:mean2013gummadi}. Denote $r(f_n, a^i_n)$ as the realized reward of agent $i$ when pulling arm $a^i_n$ in time slot $n$, where $r(\cdot)$ is the reward  function. With the bandit feedback, any agent only observes its realized reward without knowing the reward function and the population profile.  Formally, $r(\cdot)$ is a continuous function rather than a discrete binary value as in~\cite{Ref:mean2013gummadi,Ref:distributed2017maghsudi}. Apparently, continuous reward is more general, also widely adopted  in both the single-agent MAB~\cite{Ref:UCB,Ref:EXP3} and the multi-agent learning~\cite{Ref:mean2018yang,Ref:learning2017heliou}. By convention, we assume that the reward $r(f_n, a^i_n)$ is in the range $[0,1]$ which can be easily extended to other arbitrary intervals.

The state is updated after an agent obtains a realized reward. If agent $i$ observes $r(f_n,a^i_n)$,  we update its state as:
\begin{equation} \label{Eq:StateUpdate} 
s^i_{n+1} (j)= (1- \gamma_{n})s^i_n(j) + \gamma_{n} w^i_n(j), 
\end{equation}
where:
\begin{equation} \label{Eq:realizedRrd}
w^i_n(j)=
\left \{
\begin{aligned}
&r(f_n, a^i_n)  ~&& \mathrm{if}~a^i_n=j,\\
&s^i_n (j)~&&\mathrm{otherwise}.
\end{aligned}
\right.
\end{equation}
The updating rule implies that only the state to the played arm is renewed while others remain unchanged. Moreover, the stepsize $\gamma_n$ satisfies the following condition:
\begin{equation} \label{Eq:timestep}
\sum \nolimits_{n} \gamma_n = \infty,~~\sum \nolimits_{n} \gamma_n^2 < \infty.
\end{equation} 
\vspace{-12pt}
\subsection{Objective}
Define $\bm{s}_n = \left[s^1_n,s^2_n,...,s^N_n\right]$ as the state profile, i.e., the states of all agents, thus $\bm{s}_n \in [0,1]^{N \times M}$ since $r(f_n,j) \in [0,1]$. Our objective is to analyze the convergence of $\bm{s}_n $, particularly to derive the existence and uniqueness of  MFE. Meanwhile, we will also deduce the cumulative state change, which entails the regret information of the MAB problem. 

\vspace{-3pt}
\section{Existence and Uniqueness of MFE} \label{Sec:ExistEquili}
Characterizing the MFE is critical to a multi-agent system, because it can provide a guaranteed and predictable system performance~\cite{adlakha2015equilibria}. To achieve this goal,  we need to answer two fundamental questions: 1) Does MFE exist? 2) If so, does there exist only one MFE?

The evolution of  state profile $\bm{s}_n $ can be decomposed into playing and evolving processes. The playing  process is the stationary policy of Eq.~\eqref{Eq:hedge}, mapping states to population profile $\Gamma_1: \bm{s}_n \rightarrow f_n$; the evolving process amounts to the state updating of Eq.~\eqref{Eq:StateUpdate}, mapping population profile to states in turn $\Gamma_2: f_n \rightarrow \bm{s}_{n+1}$. Let $\Gamma = \Gamma_2 \circ \Gamma_1$ be the compound  mapping, so the state evolution  is  interpreted as $\Gamma: \bm{s}_n \rightarrow \bm{s}_{n+1}$. The definition of MFE  under mapping $\Gamma$ is now presented.

\begin{dfn} \label{Dn:equilibrium}
State profile $\overline{\bm{s}}$ is a MFE if $\overline{\bm{s}} = \Gamma(\overline{\bm{s}})$. 
\end{dfn}

\subsection{Existence of MFE} 
Definition~\ref{Dn:equilibrium} indicates that a MFE is indeed a \emph{fixed point under mapping $\Gamma$}. Since agents  stochastically  play actions  based on  Eq.~\eqref{Eq:hedge}, we obtain that  $\Gamma$ is a set mapping. Suppose the reward $r(f_n, j)$ is continuous in the population profile $f_n$.

\begin{theorem}  \label{Thm:fixpoint}
There exists a MFE $\overline{\bm{s}}$ satisfying  $\overline{\bm{s}} \in \Gamma(\overline{\bm{s}})$.
\end{theorem}
\begin{proof}
The state profile $\bm{s}_n$ is in a nonempty, compact, and convex set $[0,1]^{N \times M}$. The mapping $\Gamma$ maps $\bm{s}_n$ to $\bm{s}_{n+1} \in [0,1]^{N \times M}$ which is a nonempty, closed, convex subset of $[0,1]^{N \times M}$. We next prove $\Gamma$ is upper semi-continuous. 

Let $\mathcal{P}$ be the set of the population profile  $f_n$. For any  state profile $\bm{s}_n$, we have $\Gamma_1(\bm{s}_n) \in \mathcal{P}$ according to the playing process. From Eq.~\eqref{Eq:StateUpdate},  if agent $i$ plays arm $j$, then its state corresponding to this arm  after updating satisfies:
\begin{equation} \label{Eq:setupdate}
s^i_{n+1}(j) \in \bigcup \nolimits_{f_n \in \mathcal{P}} \left\{(1-\gamma_{n})s^i_n(j)+\gamma_{n}r(f_n, j) \right\}.
\end{equation}
Suppose there are arbitrary  sequences $\bm{x}_n, \bm{y}_n \in [0,1]^{N \times M}$ such that $\lim_{n \rightarrow \infty}\bm{x}_n \rightarrow \overline{\bm{x}},\lim_{n \rightarrow \infty} \bm{y}_n \rightarrow \overline{\bm{y}}$, and $\bm{y}_n \in \Gamma(\bm{x}_n)$. Given a state profile, the set $\mathcal{P}$ is determined. Since $r(f_n, j)$ is continuous in $f_n$, we can claim that $\overline{\bm{y}} \in \Gamma(\overline{\bm{x}})$ from  Eq.~\eqref{Eq:setupdate}, i.e., $\Gamma$ is upper semi-continuous.  By applying the  Kakutani fixed-point theorem~\cite{glicksberg1952further}, there exists a fixed point $\overline{\bm{s}}$, or MFE,  under the mapping $\Gamma$.
\end{proof}
Theorem~\ref{Thm:fixpoint} ensures the existence of MFE, to which the states will ultimately converge. Since there may exist multiple MFEs, it is hard to determine which MFE the states will eventually approach. Therefore, we need to further derive a unique MFE so as to achieve a more guaranteed performance. 

\subsection{Uniqueness of MFE}
Unique MFE implies there is only one fixed point. For mean field game, Lasry and Lions show that if the reward is a monotone function known by any agent, one will obtain a unique MFE~\cite{Ref:mean2007lasry}, which  however is inconsistent with the MAB problem. To handle the bandit feedback, we use the stochastic approximation~\cite{Ref:dynamics1999Benaim} to transform the discrete-time bandit game into a continuous-time ODE, and  derive the unique MFE  by proving that the ODE will  only converge to one fixed point. Therefore, we can resort to \emph{the deterministic ODE to figure out the stochastic bandit game}.

\subsubsection{Stochastic approximation}
The state evolution is a stochastic process in that each agent plays arms randomly according to the stationary policy. Considering this fact, we rewrite the state updating of Eq.~\eqref{Eq:StateUpdate}:
\begin{equation} \label{Eq:Update2}
\begin{aligned}
s^i_{n+1}(j) &= s^i_n(j) + \gamma_{n}\left(w^i_n(j) - s^i_n(j)\right)\\
&= s^i_n(j) + \gamma_{n} \left(\mathbb{E}[w^i_n(j)] - s^i_n(j) + u^i_{n}(j)\right),
\end{aligned}
\end{equation}
where $u^i_{n}(j) = w^i_n(j) - \mathbb{E}[w^i_n(j)]$. Denote $\bm{s}_t=\left[s^1_t, s^2_t,...,s^N_t\right]$ as  the state profile at continuous time $t$. In the following, we show that the state profile $\bm{s}_n$ at discrete time will \emph{asymptotically converge} to  $\bm{s}_t$, which is characterized by a deterministic ODE. Hereinafter, we use $t$ to index continuous time, and $n$ to index discrete time slot.

\begin{lemma} \label{Lem:pseudo}
When $n \rightarrow \infty$ and $t \rightarrow \infty$, the state  $s^i_n =[s^i_n(1),s^i_n(2),...,s^i_n(M)]$ of agent $i$ will asymptotically converge to $s^i_t$ specified by the following ODE:
\begin{equation} \label{Eq:ode}
\frac{ds^i_t}{dt} = \mathbb{E}[w^i_t|\bm{s}_t] -s^i_t.
\end{equation}
\end{lemma}

See Appendix~\ref{App:pseudo} for the proof. With this lemma,  if the ODE Eq.~\eqref{Eq:ode} of $\bm{s}_t$ solely converges to one fixed point, then there is a unique MFE for $\bm{s}_n$. Note that the mapping $\Gamma$ for the continuous-time state evolution now corresponds to  the ODE Eq.~\eqref{Eq:ode}. Specifically,  using Eqs.~\eqref{Eq:StateUpdate} and \eqref{Eq:realizedRrd}, the ODE of $s_t^i(j)$ is obtained:
\begin{equation}  \label{Eq:ODEsingle}
\frac{ds_t^i(j)}{dt} = \sigma(s_t^i,j) \left(r(f_t, j)-s_t^i(j)\right).
\end{equation}
We emphasize that $r(f_t, j)$ in  Eq.~\eqref{Eq:ODEsingle} is the \emph{expected reward} over the state profile $\bm{s}_t$ as the population profile  $f_t$ is mapped from $\bm{s}_t$ by the stochastic stationary policy. 

\subsubsection{Convergence to unique MFE}
Consistent with Lemma~\ref{Lem:pseudo}, we first assume the reward function is a contraction mapping to obtain the unique fixed point, and then derive the conditions for this contraction mapping. We  express the reward $r(f_t, j)$ of playing arm $j$ as $r(f(\bm{s}_t), j)$ to explicitly indicate its  dependence on $\bm{s}_t$. Also denote $\overline{\bm{s}}$ as the fixed point  for ODE Eq.~\eqref{Eq:ODEsingle} at which  the derivative is 0. 
\begin{theorem}  \label{Thm:convergence}
Suppose that the reward function $r(f(\bm{s}_t), j)$ is a $||\cdot||_{\infty}$-contraction  in the state profile $\bm{s}_t$, then the fixed point $\overline{\bm{s}}$ is  the unique MFE for the bandit game. $\overline{\bm{s}}$ is also the global attractor for the ODE Eq.~\eqref{Eq:ODEsingle}, and $\bm{s}_t$ converges to $\overline{\bm{s}}$ with exponential rate.
\end{theorem}
\begin{proof}
Because $r(f(\bm{s}_t), j)$ is a contraction mapping in $\bm{s}_t$, $r(f(\bm{s}_t), j)- s^i_t(j) =0$ in Eq.~\eqref{Eq:ODEsingle} will  have only one fixed point $\overline{\bm{s}}$. From Lemma~\ref{Lem:pseudo}, we obtain that $\overline{\bm{s}}$ is also the unique MFE for the bandit game.

Construct Lyapunov function $V(\bm{s}_t) = ||\bm{s}_t - \overline{\bm{s}}||_{\infty}$, and  assume $V(\bm{s}_t)$ gets its  maxima at $s^i_t(j) $, that is  $V(\bm{s}_t)=|s^i_t(j) - \overline{s}^i(j)|$. If $s^i_t(j) > \overline{s}^i(j)$, we have $V(\bm{s}_t)=s^i_t(j) - \overline{s}^i(j)$. Take derivative over time $t$ and use $r(f(\overline{\bm{s}}), j) = \overline{s}^i(j)$:
\begin{equation} \label{Eq:Lyapunov}
\scalebox{0.98}{$\begin{aligned}
&\frac{dV(\bm{s}_t)}{dt} =\sigma(s^i_t,j)\left[r(f(\bm{s}_t), j)-s^i_t(j) - r(f(\overline{\bm{s}}), j)+\overline{s}^i(j)\right]\\
& \le  \sigma(s^i_t,j)\left[C_1||\bm{s}_t - \overline{\bm{s}}||_{\infty} - |s^i_t(j)-\overline{s}^i(j)| \right]\\
& = \sigma(s^i_t,j) (C_1 -1) |s^i_t(j)-\overline{s}^i(j)|\\
& = \sigma(s^i_t,j) (C_1 -1) V(\bm{s}_t),
\end{aligned}$}
\end{equation}
where $C_1 < 1$ because of the $||\cdot||_{\infty}$-contraction mapping. As a result, $\frac{dV(\bm{s}_t)}{dt} \le 0$ and $\frac{dV(\bm{s}_t)}{dt} = 0$ only at $\overline{\bm{s}}$. If  $s^i_t(j) \le \overline{s}^i(j)$, following the same approach, we can also obtain $\frac{dV(\bm{s}_t)}{dt} \le \sigma(s^i_t,j) (C_1 -1) V(\bm{s}_t)$. Therefore, the fixed point $\overline{\bm{s}}$ is a global attractor for the ODE.

As Lyapunov function satisfies $\frac{dV(\bm{s}_t)}{dt} \le \sigma(s^i_t,j) (C_1 -1) V(\bm{s}_t)$, which implies $\frac{{dV(\bm{s}_t)}}{V(\bm{s}_t)} \le \sigma(s^i_t,j) (C_1 -1) dt$, we have 
$V(\bm{s}_t) \le C_2\mathrm{Exp}\left[\sigma(s^i_t,j) (C_1 -1)t\right]$ where $C_2$ is a constant. Therefore, $\bm{s}_t$ converges to $\overline{\bm{s}}$ exponentially fast.
\end{proof}

According to  Lemma~\ref{Lem:pseudo} and Theorem~\ref{Thm:convergence}, we know that $\bm{s}_n \rightarrow \overline{\bm{s}}$ when $n \rightarrow \infty$. Next, we discuss the convergence rate of the discrete-time $\bm{s}_n$. Recall from the stepsize $\gamma_n$ in Eq.~\eqref{Eq:timestep}, we can set $\gamma_n = 1/(n+1)^{\alpha}, \alpha \in \left(1/2,1\right]$.
\begin{theorem} \label{Thm:DisConver}
Suppose that the reward function $r(f(\bm{s}_n), j)$ is a $||\cdot||_{\infty}$-contraction in the state profile $\bm{s}_n$. Denote the distance $e_n = ||\bm{s}_n - \overline{\bm{s}}||_{\infty}$:

1) if $\alpha \in \left(\frac{1}{2},1\right)$, given $n= \Omega \left( \left(\frac{ \ln \frac{1}{\delta \epsilon} }{\epsilon^2} \right)^{\frac{1}{\alpha}} + \left(\ln \frac{1}{\epsilon}\right)^{\frac{1}{1-\alpha}}\right)$, then $e_n \le \epsilon$ with  probability at least  $1-\delta$;

2) if $\alpha=1$,  given $n= \Omega \left((2+\Psi)^{\ln \frac{1}{\epsilon}} \frac{ \ln \frac{1}{\delta \Psi \epsilon} }{\Psi^2 \epsilon^2}\right)$, then $e_n \le \epsilon$ with probability at least $1-\delta$ for any positive constant $\Psi$.
\end{theorem}
\begin{proof}
 Based on Theorems 2 and 3 in~\cite{Ref:learning2003Dar}, we only need to prove $||\mathbb{E}[\bm{w}_n(\bm{s}_n)] - \overline{\bm{s}}||_{\infty} \le C ||\bm{s}_n - \overline{\bm{s}}||_{\infty}$ with $C \in [0,1)$, to draw the two conclusions.  Assume $||\mathbb{E}[\bm{w}_n] - \overline{\bm{s}}||_{\infty} = |\mathbb{E}[w^i_n(j)] - \overline{s}^i(j)|$. According to Eq.~\eqref{Eq:realizedRrd} and the fact $\overline{\bm{s}}$ satisfies $r(f(\overline{\bm{s}}), j) = \overline{s}^i(j)$, we have:
\begin{equation} \label{Eq:ConConver}
\begin{aligned}
&||\mathbb{E}[\bm{w}_n] - \overline{\bm{s}}||_{\infty}=  |\mathbb{E}[w^i_n(j)] - \overline{s}^i(j)|\\
&=|\sigma(s_n^i,j)r(f(\bm{s}_n), j) + (1-\sigma(s_n^i,j))s_n^i(j) - \overline{s}^i(j)|\\
& \le |\sigma(s_n^i,j)(r(f(\bm{s}_n), j)-r(f(\overline{\bm{s}}), j)|\\
&+ |(1-\sigma(s_n^i,j))(s_n^i(j)- \overline{s}^i(j))| \\
& \le \sigma(s_n^i,j)C_1 ||\bm{s}_n- \overline{\bm{s}}||_{\infty} + (1-\sigma(s_n^i,j)) |s_n^i(j)- \overline{s}^i(j)|\\
&\le \sigma(s_n^i,j)C_1 ||\bm{s}_n- \overline{\bm{s}}||_{\infty} + (1-\sigma(s_n^i,j)) ||\bm{s}_n- \overline{\bm{s}}||_{\infty}\\
& \le C ||\bm{s}_n- \overline{\bm{s}}||_{\infty}.
\end{aligned}
\end{equation}
The second inequality holds as $r(f(\bm{s}_n), j)$ is a contraction, and the third inequality is from $||\cdot||_{\infty}$-definition. 
\end{proof}
The difference in the convergence rate characterizations in Theorems~\ref{Thm:convergence} and~\ref{Thm:DisConver} is because  the state profile $\bm{s}_t$ is a deterministic process, while  the state profile $\bm{s}_n$ is a stochastic process. Hence, Theorem~\ref{Thm:DisConver} uses a probabilistic description for the convergence rate of  $\bm{s}_n$ to the unique MFE $\overline{\bm{s}}$. 

\subsubsection{Contraction mapping condition}
The foundation  of the unique MFE  lies in that the reward function  $r(f(\bm{s}_t), j)$ is a contraction mapping in the state profile $\bm{s}_t$.  Suppose the  reward function $r(f_n,j)$ is $\theta$-Lipschitz  continuous in the population profile $f_n$ with regard to $||\cdot||_1$-norm: $|r(f_n, j) - r(f'_n, j)| \le \theta ||f_n -f'_n||_1$. Combining with the idea in~\cite{Ref:a2010cominetti} and the mean field  model,  we now characterize the contraction mapping condition.
\begin{theorem} \label{Thm:contraction}
If parameters $\beta,\eta$ in Eq.~\eqref{Eq:hedge} and $\theta$ satisfy the condition $4\theta(1-\eta)\beta <1$, then the reward  function $r(f(\bm{s}_t), j)$ is a $||\cdot||_{\infty}$-contraction in the state profile $\bm{s}_t$.
\end{theorem}
\begin{proof}
Let $\bm{s}_a$ and $\bm{s}_b$ be two state profiles, and define a sequence $\bm{A}_k = \left[s_a^1, s_a^2,...,s_a^k,s_b^{k+1},...,s_b^N\right], k = 1,2,...,N$ with the first $k$ elements from $\bm{s}_a$ and the rest from $\bm{s}_b$. Denote $A_k^i$ as the $i$-th  state in $\bm{A}_k$, and $\bm{a}_k$ as the played arms following  the stationary policy  $\sigma(A_k^i)$. Let $\bm{a}_k^{-i}$ be the arm set except agent $i$, and $f(\bm{a}_k)$  or $f(\bm{A}_k)$ be the population profile.  When  $a_k^i = j$, $|r(f(\bm{s}_a), j) - r(f(\bm{s}_b), j)| = |\sum_{k=1}^{N} [r(f(\bm{A}_k), j) - r(f(\bm{A}_{k-1}), j)]| \le \sum_{k=1}^N |r(f(\bm{A}_k), j) - r(f(\bm{A}_{k-1}), j)|=  \sum_{k=1}^N |r(f(\bm{a}_k), j) - r(f(\bm{a}_{k-1}), j)|$.

Let $\triangle_k = r(f(\bm{a}_k), j) - r(f(\bm{a}_{k-1}), j)$, so we can express $\triangle_k = \Delta r_k \cdot (\sigma(s_a^k)-\sigma(s_b^k))$ where $\Delta r_k $ is a $M$-length vector with the $l$-th element $\Delta r_k(l)$ as $\Delta r_k(l) = \sum_{\bm{z} \in \bm{a}_k^{-i}, z^k = l} r(f(\{j,\bm{z}\}), j) \prod_{m \ne k,i} \sigma(A_k^m,z^m)$.
Note that $\Delta r_k(h) \bm{1} \cdot  (\sigma(s_a^k)-\sigma(s_b^k)) = 0, \forall h \in \mathcal{M}$, where $\bm{1}$ is a $M$-length vector with each element equal to 1. Therefore, we can rewrite $\triangle_k = (\Delta r_k - \Delta r_k(h) \bm{1}) \cdot  (\sigma(s_a^k)-\sigma(s_b^k))$. If $\bm{z}_a, \bm{z}_b \in \bm{a}_k^{-i}$ with only the $k$-th arm being different: $z_a^k =l, z_b^k=h$, we obtain $||f(\{j,\bm{z}_a\}) - f(\{j,\bm{z}_b\})||_1 \le \frac{2}{N}$ from Eq.~\eqref{Eq:Population}. As $r(f_n,j)$ is $\theta$-Lipschitz continuous, then:
\begin{equation}\label{Eq:PopulationLip}
|\Delta r_k(l) - \Delta r_k(h)| \le  \Bigl|\sum_{\bm{z} \in \bm{a}^{-i}}  \prod_{m \ne k,i} \sigma(A_k^m,z^m) \Bigr| \frac{2 \theta}{N} \le \frac{2 \theta}{N}.
\end{equation}
As a result, we obtain:
\begin{equation} \label{Eq:rbound}
|r(f(\bm{s}_a), j) - r(f(\bm{s}_b), j)| \le \frac{2 \theta}{N} \sum \nolimits_{k=1}^N ||\sigma(s_a^k)-\sigma(s_b^k)||_1.
\end{equation}

From mean value theorem, there is a $x \in [s_a^k, s_b^k]$ such that:
\begin{equation} \label{Eq:Mean}
\sigma(s_a^k, j)-\sigma(s_b^k, j) = \nabla \sigma(x,j) \cdot  (s_a^k - s_b^k).
\end{equation}
Hence, $|\sigma(s_a^k, j)-\sigma(s_b^k, j)| \le ||\nabla \sigma(x,j)||_1 ||s_a^k - s_b^k||_{\infty}$ and $||\sigma(s_a^k)-\sigma(s_b^k)||_1 \le \sum_{j=1}^M ||\nabla \sigma(x,j)||_1 ||s_a^k - s_b^k||_{\infty}$.  Considering the stationary policy of Eq.~\eqref{Eq:hedge}, we have $\frac{d\sigma(x,j)}{dx(l)} = (1-\eta)\beta \sigma(x,j) (\mathbbm{1}_{\{j = l\}} - \sigma(x,l))$.
Therefore:
\begin{equation} \label{Eq:Norm1Deri}
\begin{aligned}
&||\nabla \sigma(x,j)||_1  = (1-\eta)\beta \sigma(x,j) \sum \nolimits_{l=1}^{M} |(\mathbbm{1}_{\{j = l\}} - \sigma(x,l))| \\
&= 2(1-\eta)\beta \sigma(x,j)(1-\sigma(x,j)) \le 2(1-\eta)\beta \sigma(x,j).\\
\end{aligned}
\end{equation}
Combining with Eqs.~\eqref{Eq:Mean} and~\eqref{Eq:Norm1Deri}, we attain: 
\begin{equation}
\begin{aligned}
||\sigma(s_a^k)-\sigma(s_b^k)||_1 & \le \sum_{j=1}^M  2(1-\eta)\beta \sigma(x,j) ||s_a^k - s_b^k||_{\infty}\\
& \le 2(1-\eta)\beta ||\bm{s}_a - \bm{s}_b||_{\infty}.
\end{aligned}
\end{equation}
In line with Eq.~\eqref{Eq:rbound},  we present the final result:
\begin{equation}
|r(f(\bm{s}_a), j) - r(f(\bm{s}_b), j)| \le 4\theta(1-\eta)\beta  ||\bm{s}_a - \bm{s}_b||_{\infty}.
\end{equation}
\end{proof}

The condition $4\theta(1-\eta)\beta <1$ in Theorem~\ref{Thm:contraction} is  a little  stringent for two reasons. First, the reward function $r(f(\bm{s}_t), j)$ depends on \emph{all $M$ elements} of  $f(\bm{s}_t)$. Second, $r(f(\bm{s}_t), j)$ is \emph{non-linear} in $f(\bm{s}_t)$, so calculating the expected reward needs multiple scaling operations. In fact, the reward  $r(f(\bm{s}_t), j)$ of playing arm $j$ is often only impacted by the number of agents who select arm $j$~\cite{Ref:mean2013gummadi}, such as an agent playing arm $j$ only competes with those making the same choice in a resource competition game.  Hence, we assume $r(f(\bm{s}_t), j)$ merely depends on the $j$-th element, denoted as $f(\bm{s}_t,j)$. Besides, we further presume the reward is a linear function in the population profile. With these two assumptions, we recharacterize a relaxed contraction condition.

\begin{corollary} \label{Cor:LinearContr}
If the reward $r(f(\bm{s}_t,j), j)$ is a $\theta$-Lipschitz continuous linear function in the $j$-th element $f(\bm{s}_t,j)$, then $r(f(\bm{s}_t,j), j)$ is a $||\cdot||_{\infty}$-contraction in the state profile $\bm{s}_t$ under the condition $\frac{\theta(1-\eta)\beta}{2} < 1$, where $\beta, \eta$ are from Eq.~\eqref{Eq:hedge}.
\end{corollary}
\begin{proof}
From Eqs.~\eqref{Eq:hedge}-\eqref{Eq:Population}, the expected population profile,  with a slight abuse of notations, is $f(\bm{s}_t,j)=\mathbb{E} [\sum_{i=1}^N \frac{\mathbbm{1}_{\{a_t^i = j\}}}{N}]=  \frac{\sum_{i=1}^N \sigma(s_t^i,j)}{N}$.
Since  $r(f(\bm{s}_t,j), j)$ is $\theta$-Lipschitz continuous in $f(\bm{s}_t,j)$, we obtain $|\frac{d r(f(\bm{s}_t,j), j)}{d f(\bm{s}_t,j)}| \le \theta$. For two state profiles $\bm{s}_a, \bm{s}_b$, there is a $\bm{x} \in [\bm{s}_a, \bm{s}_b]$ satisfying $r(f(\bm{s}_a,j), j) - r(f(\bm{s}_b,j), j) =  \nabla r(f(\bm{x},j), j) \cdot (\bm{s}_a-\bm{s}_b)$ based on  mean value theorem. Hence, $|r(f(\bm{s}_a,j),j) - r(f(\bm{s}_b,j), j)| \le ||\nabla r(f(\bm{x},j), j)||_1||\bm{s}_a-\bm{s}_b||_{\infty}$.

Considering   $r(f(\bm{x},j), j)$ is linear in  $f(\bm{x},j)$, we have $\frac{d r(f(\bm{x},j), j)}{d x^i(l)} = \frac{d r(f(\bm{x},j), j)}{d f(\bm{x},j)} \frac{d f(\bm{x},j)}{d x^i(l)}, \forall i \in \mathcal{N}, l \in \mathcal{M}$. As  $|\frac{d r(f(\bm{x},j), j)}{d f(\bm{x},j)}| \le \theta$, we only need to handle $\frac{d f(\bm{x},j)}{d x^i(l)}$.  Moreover, $\frac{d f(\bm{x},j)}{d x^i(l)} = \frac{1}{N} \frac{d\sigma(x^i,j)}{dx^i(l)} = \frac{(1-\eta)\beta \sigma(x^i,j)(\mathbbm{1}_{\{j=l\}} - \sigma(x^i,l))}{N}$.

Consequently, we acquire the following result:
\begin{equation} \label{Eq:1normr}
\begin{aligned}
& ||\nabla r(f(\bm{x},j), j)||_1  = \sum \nolimits_{i=1}^{N} \sum \nolimits_{l=1}^M \left| \frac{d r(f(\bm{x},j), j)}{d x^i(l)} \right| \\
& \le \frac{\theta(1-\eta)\beta}{N} \sum \nolimits_{i=1}^{N} \sum \nolimits_{l=1}^M |\sigma(x^i,j)(\mathbbm{1}_{\{j=l\}} - \sigma(x^i,l))|\\
& =  \frac{\theta(1-\eta)\beta}{N} \sum \nolimits_{i=1}^{N} 2\sigma(x^i,j)(1-\sigma(x^i,j))\\
& \le \frac{\theta(1-\eta)\beta}{N} \sum \nolimits_{i=1}^{N} 2 \times \frac{1}{4} \\
&=  \frac{\theta(1-\eta)\beta}{2},
\end{aligned}
\end{equation}
where  the second inequality  is because $y(1-y) \le \frac{1}{4}, \forall y \in [0,1]$. Finally, it results in:
\begin{equation}
\scalebox{0.96}{$|r(f(\bm{s}_a,j), j) - r(f(\bm{s}_b,j), j)| \le  \frac{\theta(1-\eta)\beta}{2} ||\bm{s}_a-\bm{s}_b||_{\infty}.$}
\end{equation}
\end{proof}

Comparing Theorem~\ref{Thm:contraction} and Corollary~\ref{Cor:LinearContr}, the factor 4 is now reduced to $\frac{1}{2}$. Therefore, parameters $\theta, \beta, \eta$ can take much broader range of values to ensure the contraction mapping.

\section{State Change and Model Extension} \label{Sec:StateHeter}
In this section, we derive the cumulative state change to infer the regret information and extend the mean field model to demonstrate its effectiveness in various scenarios.

\subsection{Cumulative State Change}
After state updating according to Eq.~\eqref{Eq:StateUpdate}, the state change is $\Delta s^i_n = \gamma_{n} (w^i_n - s^i_n)$. Let $\sigma(s^i_n) \cdot \Delta s^i_n$ be the inner product between the simplex  $\sigma(s^i_n)$ in Eq.~\eqref{Eq:hedge} and $\Delta s^i_n$. Also denote $(\Delta s^i_n)^2$ as the vector with square on each element of $\Delta s^i_n$.  The following theorem provides the cumulative state change.

\begin{theorem} \label{Thm:cumulative}
Denote $\bm{1}$ as a $M$-length vector with each element equal to 1.  For agent $i$ and an arbitrary arm $j$, we have:
\vspace{-12pt}
\begin{equation} \label{Eq:CumuChange}\scalebox{0.96}{$
\begin{aligned}
&\beta s^i_0(j) + \sum \nolimits_{n=0}^{K} \beta \Delta s^i_n(j) - \ln \Bigl(\sum \nolimits_{j=1}^{M} \mathrm{Exp}\left(\beta s^i_0(j)\right)\Bigr) \le \\
&\sum \nolimits_{n=0}^{K}\Bigl[ \frac{\beta (\sigma(s^i_n)-\frac{\eta}{M}\bm{1}) \cdot \Delta s^i_n }{1-\eta}+  \frac{(e-2)\beta^2 \sigma(s^i_n) \cdot (\Delta s^i_n)^2}{1-\eta}\Bigr].
\end{aligned}$}
\vspace{-2pt}
\end{equation}
\end{theorem}
See Appendix~\ref{App:cumulative} for the proof. Cumulative state change entails the regret of arm playing, i.e., $\max_{j}\sum_{n=0}^{T-1} (r(f_n, j) -\mathbb{E}[r(f_n, a^i_n)])$. Specifically, if regarding the left-hand side of  Eq.~\eqref{Eq:CumuChange} as $\max_{j} r(f_n, j)$, and expanding the right-hand side via $r(f_n, a^i_n)$,  one can obtain a theoretically loose bound of a scaled regret, which serves as a \emph{tradeoff between regret minimization and system stability}. We will show in the evaluations that the stationary policy actually has a small empirical regret.  

\emph{Remark}:  The bandit game will not converge to an equilibrium  when applying the traditional MAB algorithms which may have tight-bounded regret, like UCB~\cite{Ref:UCB} and EXP3~\cite{Ref:EXP3}. This is because they need to model an agent state as the cumulative reward. Hence, the state will consistently ``increase'' since realized rewards are positive, so that the system is unstable. This is also the underlying reason why works~\cite{Ref:mean2013gummadi,Ref:distributed2017maghsudi,Ref:intelligent2017zhao} have to assume a state regeneration to achieve system  stability.
\vspace{-3pt}
\subsection{Heterogeneous Learning Parameter} \label{Sec:heter}
So far, we have analyzed the bandit game when the stationary policy adopts homogeneous learning parameters, namely uniform $\beta, \eta$ in Eq.~\eqref{Eq:hedge}. Next, we explore the heterogeneous situation where $\beta, \eta$ could vary for different agents. 

For agent $i$, parameter $\beta^i$ will keep unchanged, while $\eta_n^i$ is diminishing and satisfies $\lim_{n \rightarrow \infty} \eta_n^i = 0$ to give less weight to the random choice  when the reward information is accurately learned. Therefore, the stationary policy changes to:
\vspace{-5pt}
\begin{equation} \label{Eq:hedge2}
\sigma(s^i_n, j) =  (1-\eta^i_n) \frac{\mathrm{Exp}\left(\beta^i s^i_n(j)\right)}{\sum_{k=1}^M \mathrm{Exp} \left(\beta^i s^i_n(k)\right)} + \frac{\eta^i_n}{M}.
\end{equation}
Except the playing policy, the updating rule is the same, as described in Eq.~\eqref{Eq:StateUpdate}. Following a similar analysis in the  homogeneous situation, it can be asserted that the results in Theorems~\ref{Thm:fixpoint}-\ref{Thm:DisConver} and Lemma~\ref{Lem:pseudo} still hold. As for the contraction mapping conditions in Theorem~\ref{Thm:contraction} and Corollary~\ref{Cor:LinearContr}, they turn out to be $4\theta \beta_{\max} <1$ and $\frac{\theta \beta_{\max}}{2} <1$, respectively, where $\beta_{\max} = \max \{\beta^i| i\in \mathcal{N}\}$. The cumulative state change in Theorem~\ref{Thm:cumulative} is specified by replacing $\beta,\eta$ with $\beta^i,\eta_n^i$.

\subsection{Overlapping Arms}
Another extension is about overlapping arms.  Mathematically, there are $M$ arms, and each agent $i$ pulls an arm from a subset $\mathcal{M}^i, \mathcal{M}^i \subseteq \mathcal{M}$ with $M^i = |\mathcal{M}^i|$. Overlapping means there exist $i,k$ such that $\mathcal{M}^i \cap \mathcal{M}^k \ne \emptyset$. The conclusion is: previous results still hold after we make several  adjustments. The proofs are the same, thus we skip them to save space.

One adjustment is to  choose arms from  $\mathcal{M}^i$ for agent $i$:
\begin{equation} \label{Eq:laphedge}
\vspace{-1pt}
\sigma(s^i_n,j) = (1-\eta) \frac{\mathrm{Exp}\left(\beta s^i_n(j)\right)}{\sum_{k \in \mathcal{M}^i} \mathrm{Exp} \left(\beta s^i_n(k)\right)} + \frac{\eta}{M^i}.
\vspace{-1pt}
\end{equation}
The state $s_n^i$ is now a $M^i$-length vector, and the state profile $\bm{s}_n$ is in $[0,1]^{\sum_{i \in \mathcal{N}}M^i}$. Other parameters, especially the cardinality  of sets, are adapted accordingly. Similar to Section~\ref{Sec:ExistEquili}, we could obtain  the existence and uniqueness of MFE.

\section{Performance Evaluation}\label{Sec:Eval}
In this section, we carry out the evaluations where results are smoothed via \emph{LOWESS} in Python for better exhibition. 

\subsection{Existence and Uniqueness of MFE}
\emph{Reward function.} We consider a competition game where agents compete for different types of resource~\cite{yang2018mean,hanif2015mean}, also regarded as arms. The reward  $r(f_n, j)$ is a non-linear decreasing function in the $j$-th element of population profile $f_n(j)$~\cite{Ref:mean2013gummadi}:
\begin{equation} \label{Eq:general} 
r(f_n, j) = \frac{1}{1+ \theta(j) f_n(j)},
\end{equation}
where $\theta(j) \in [0.8\theta,\theta], \forall j \in \mathcal{M}$. One can verify that $r(f_n, j)$ is in $[0,1]$ and satisfies $\theta$-Lipschitz continuity. Besides, the stepsize  in Eq.~\eqref{Eq:StateUpdate} is $\gamma_n = 1/(n+1)$.

\begin{figure}[t]
\centering
\subfloat[$N=50$]{\includegraphics[width=0.325\linewidth,height=0.67in]{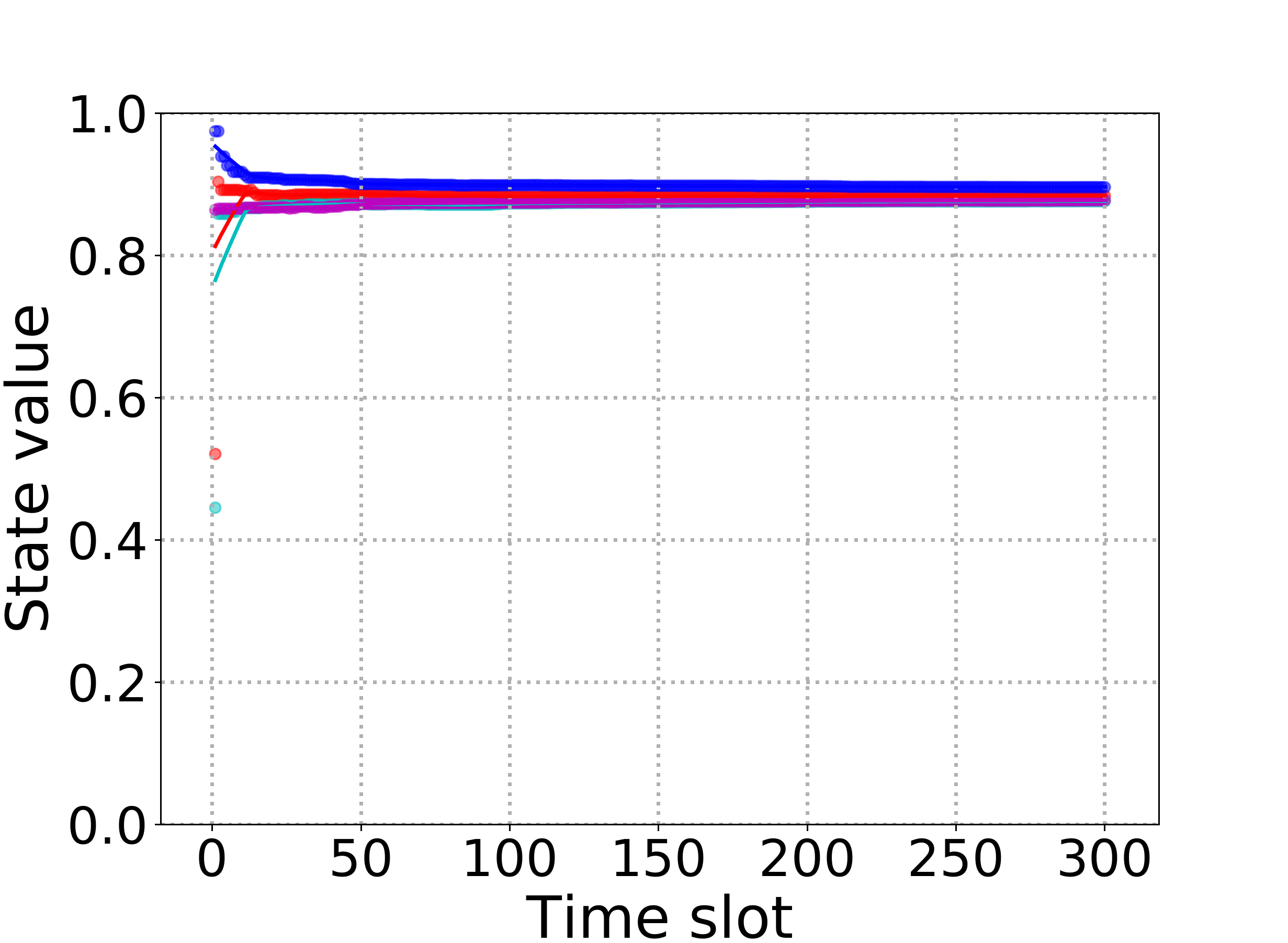}} \ 
\subfloat[$N=100$]{\includegraphics[width=0.325\linewidth,height=0.67in]{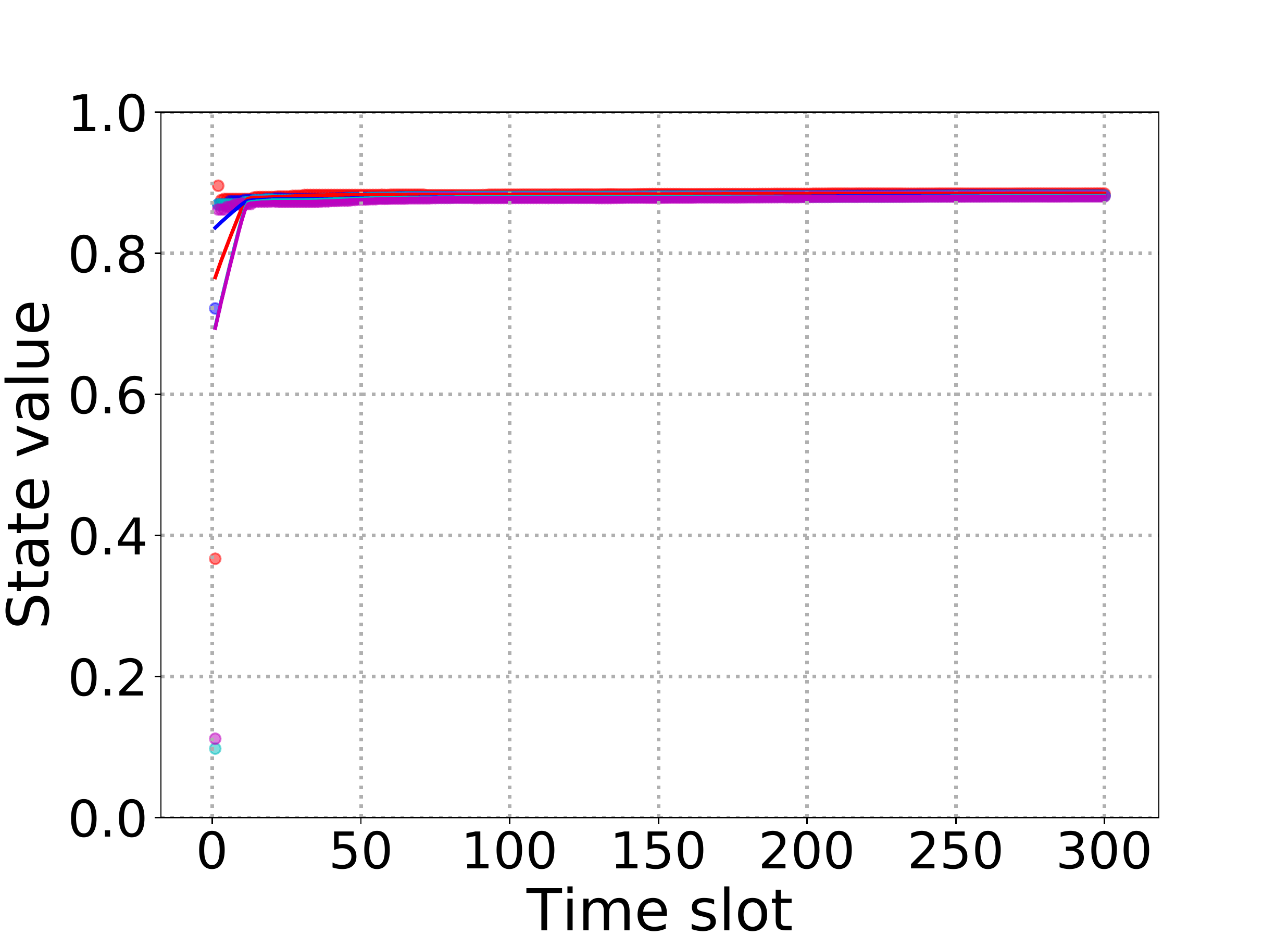}} \
\subfloat[$N=200$]{\includegraphics[width=0.325\linewidth,height=0.67in]{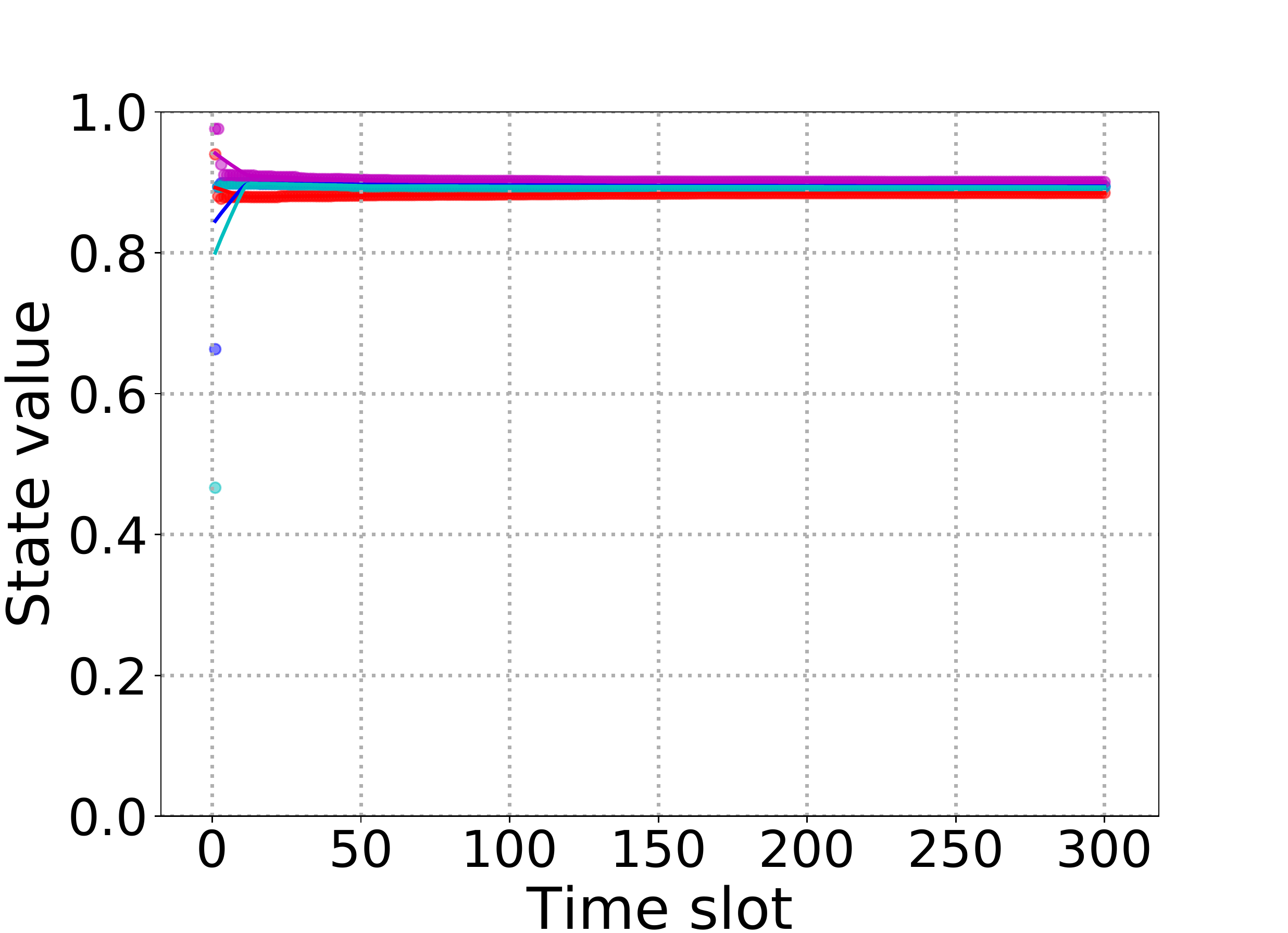}} \\
\vspace{-6pt}
\caption{Contraction mapping: state evolution.}
\label{Fig:GConState}
\vspace{-16pt}
\end{figure}

\begin{figure}[t]
\centering
\subfloat[$N=50$]{\includegraphics[width=0.325\linewidth,height=0.67in]{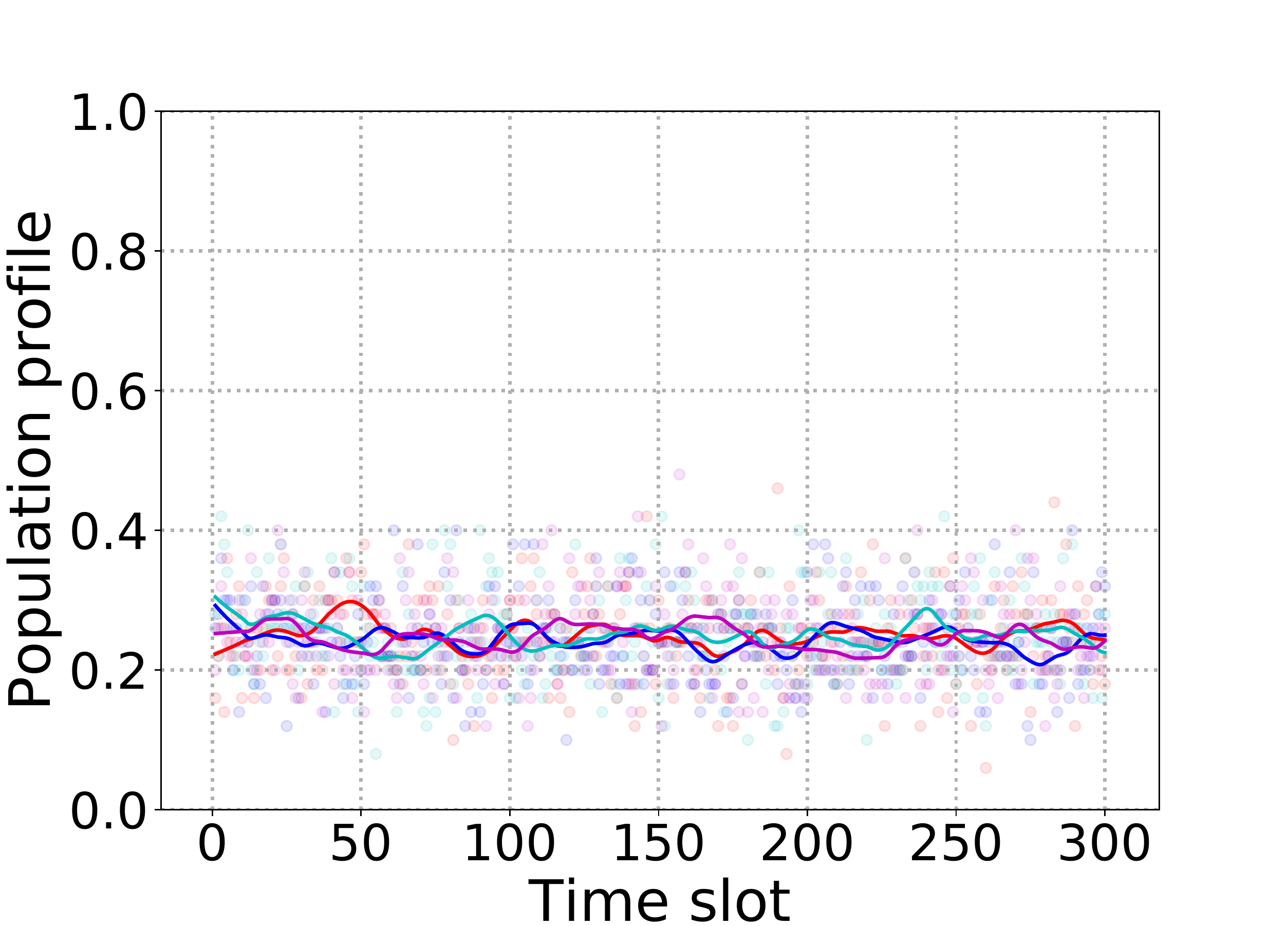}} \ 
\subfloat[$N=100$]{\includegraphics[width=0.325\linewidth,height=0.67in]{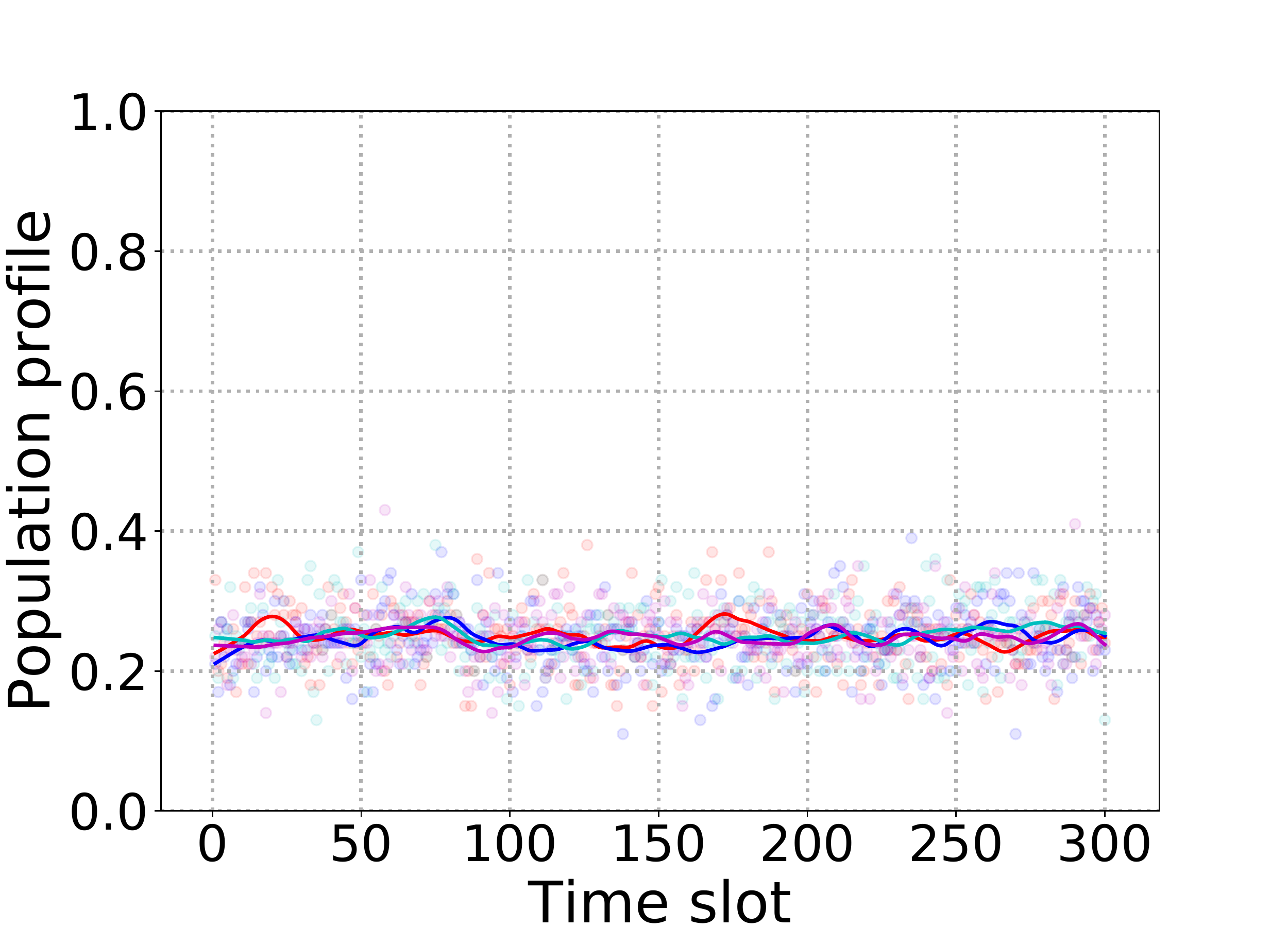}} \
\subfloat[$N=200$]{\includegraphics[width=0.325\linewidth,height=0.67in]{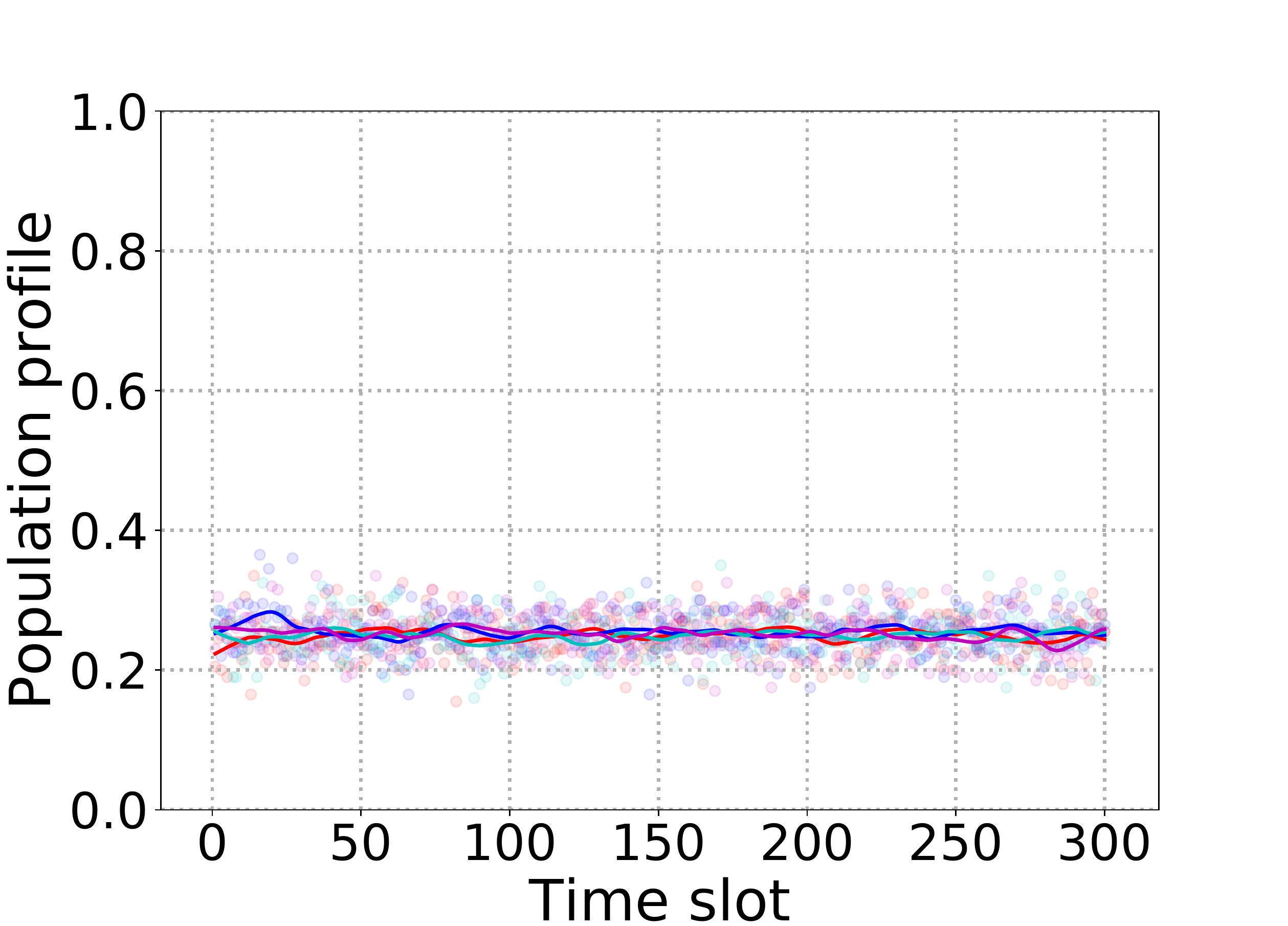}} \\
\vspace{-6pt}
\caption{Contraction mapping: population profile evolution.}
\label{Fig:GConPop}
\vspace{-15pt}
\end{figure}

\begin{figure}[t]
\centering
\subfloat[$N=50$]{\includegraphics[width=0.325\linewidth,height=0.67in]{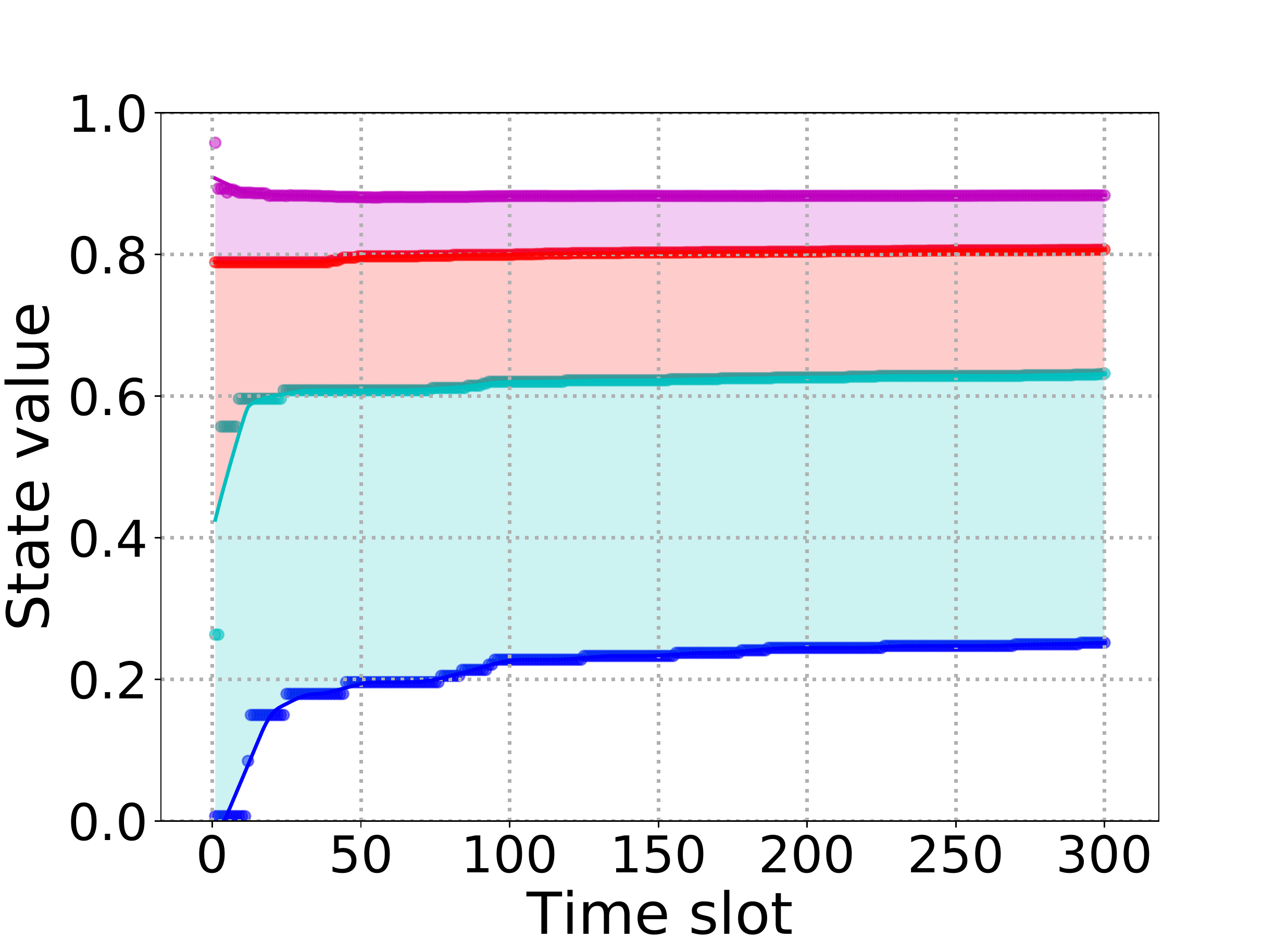}} \ 
\subfloat[$N=100$]{\includegraphics[width=0.325\linewidth,height=0.67in]{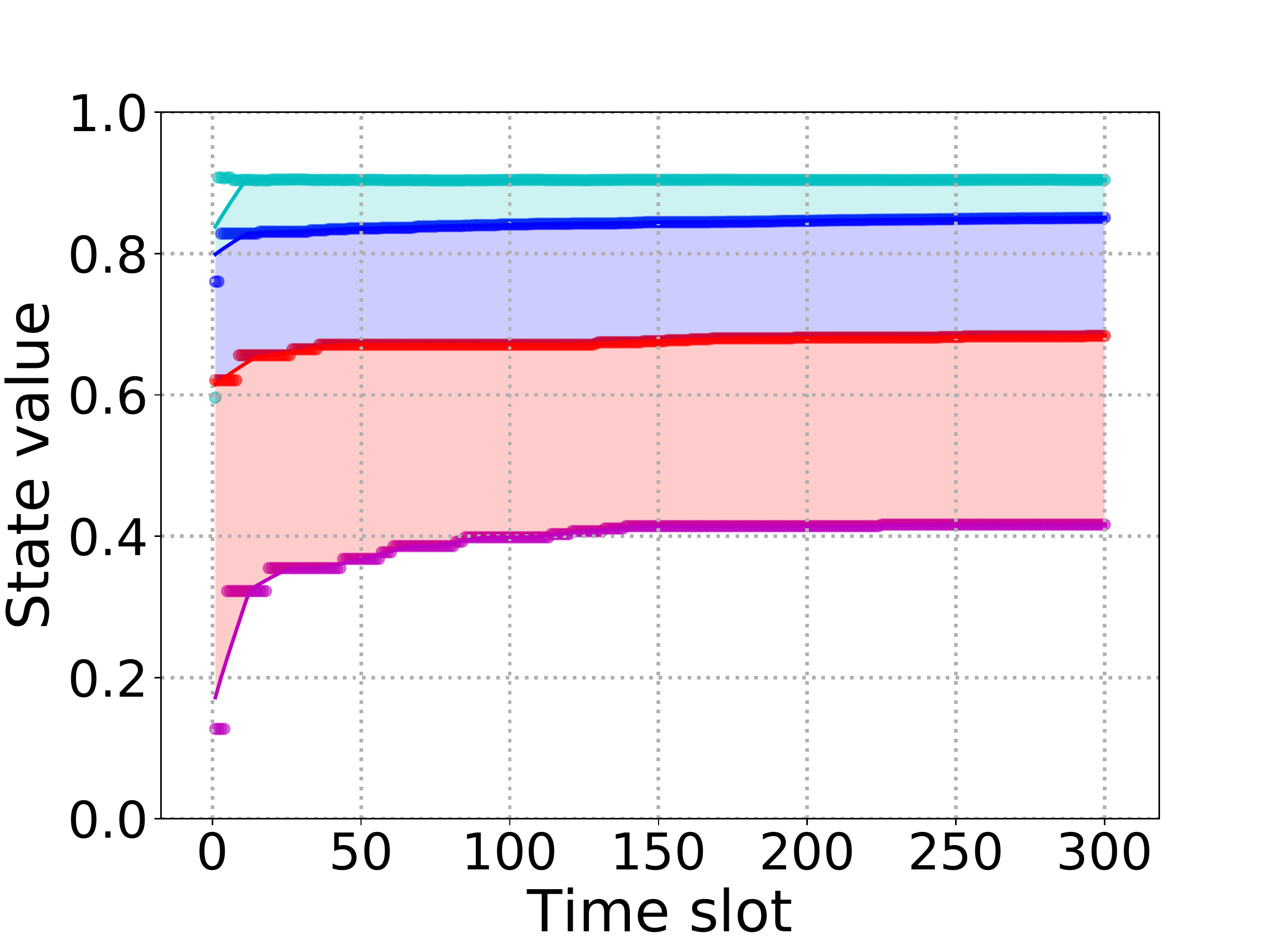}} \
\subfloat[$N=200$]{\includegraphics[width=0.325\linewidth,height=0.67in]{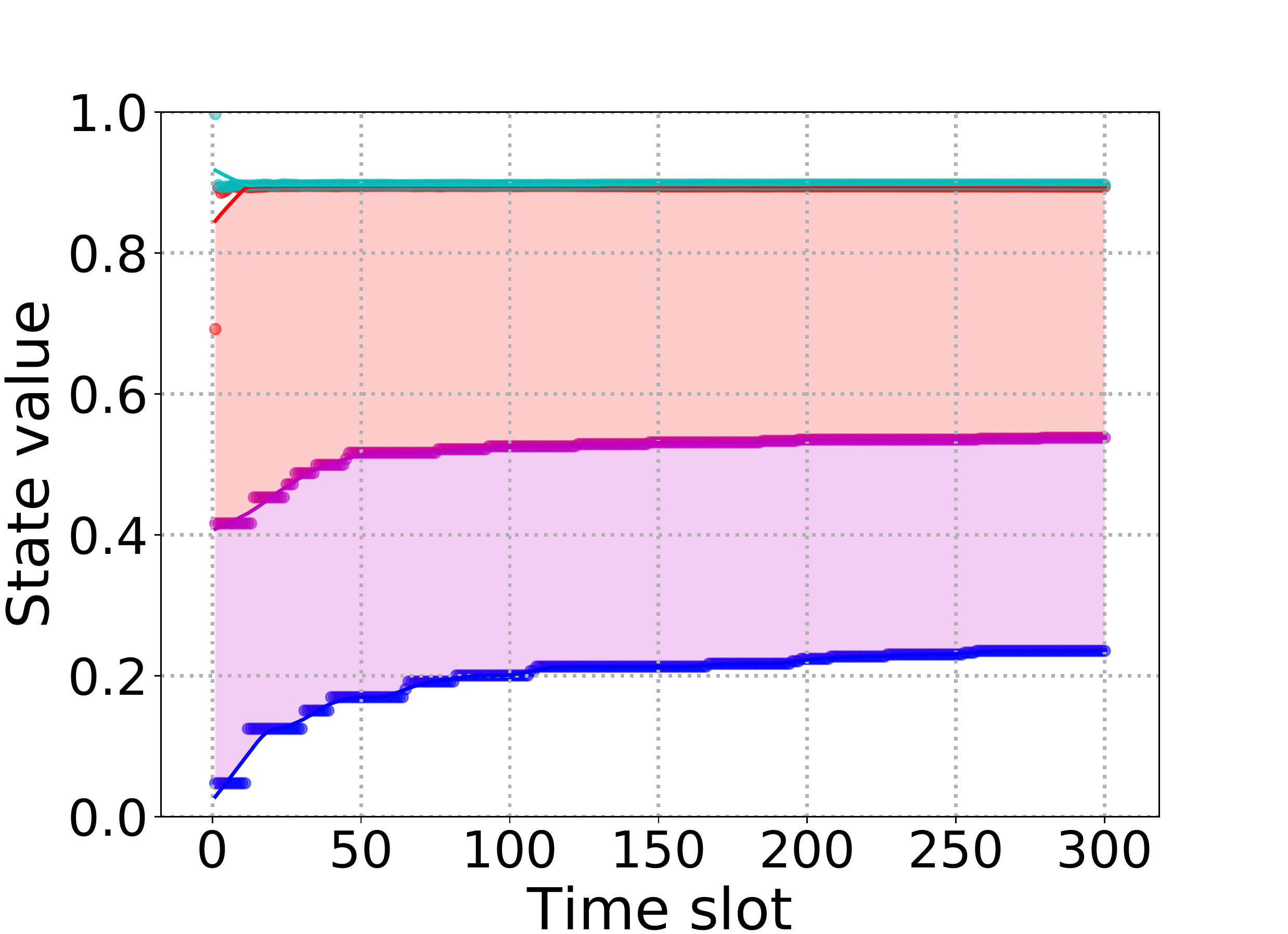}} \\
\vspace{-6pt}
\caption{Non-contraction mapping: state evolution.}
\label{Fig:GNConState}
\vspace{-16pt}
\end{figure}

\begin{figure}[t]
\centering
\subfloat[$N=50$]{\includegraphics[width=0.325\linewidth,height=0.67in]{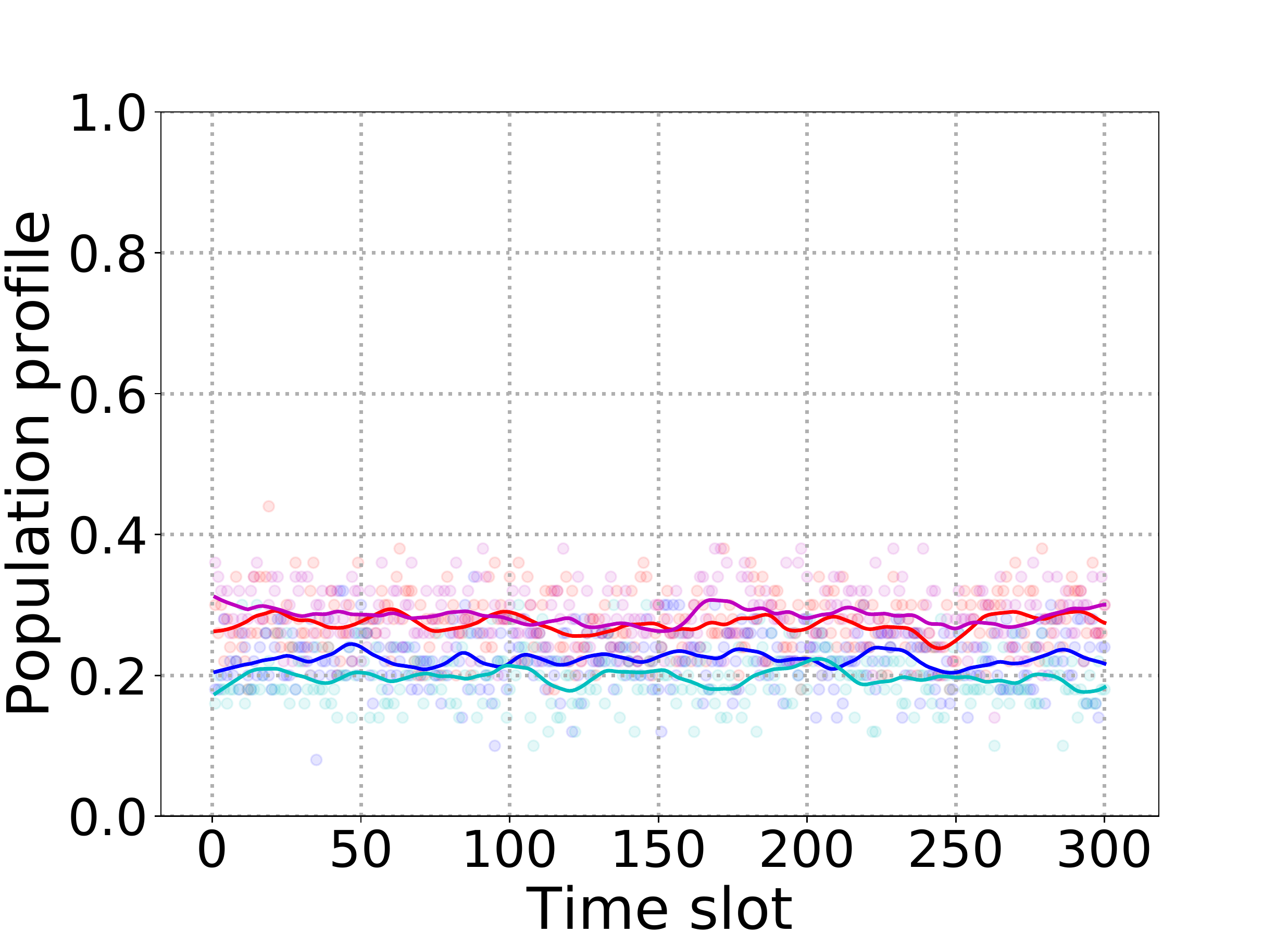}} \ 
\subfloat[$N=100$]{\includegraphics[width=0.325\linewidth,height=0.67in]{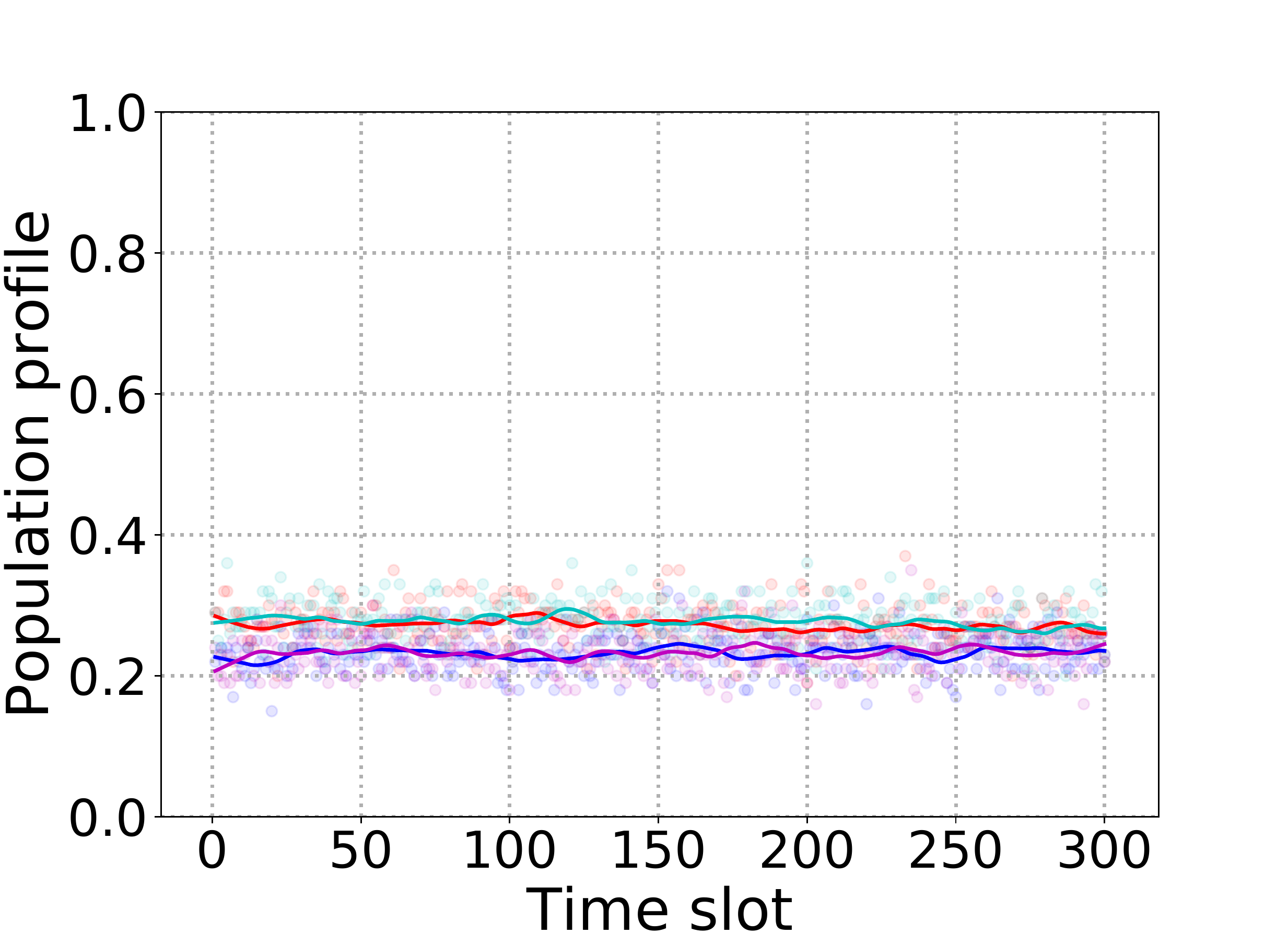}} \
\subfloat[$N=200$]{\includegraphics[width=0.325\linewidth,height=0.67in]{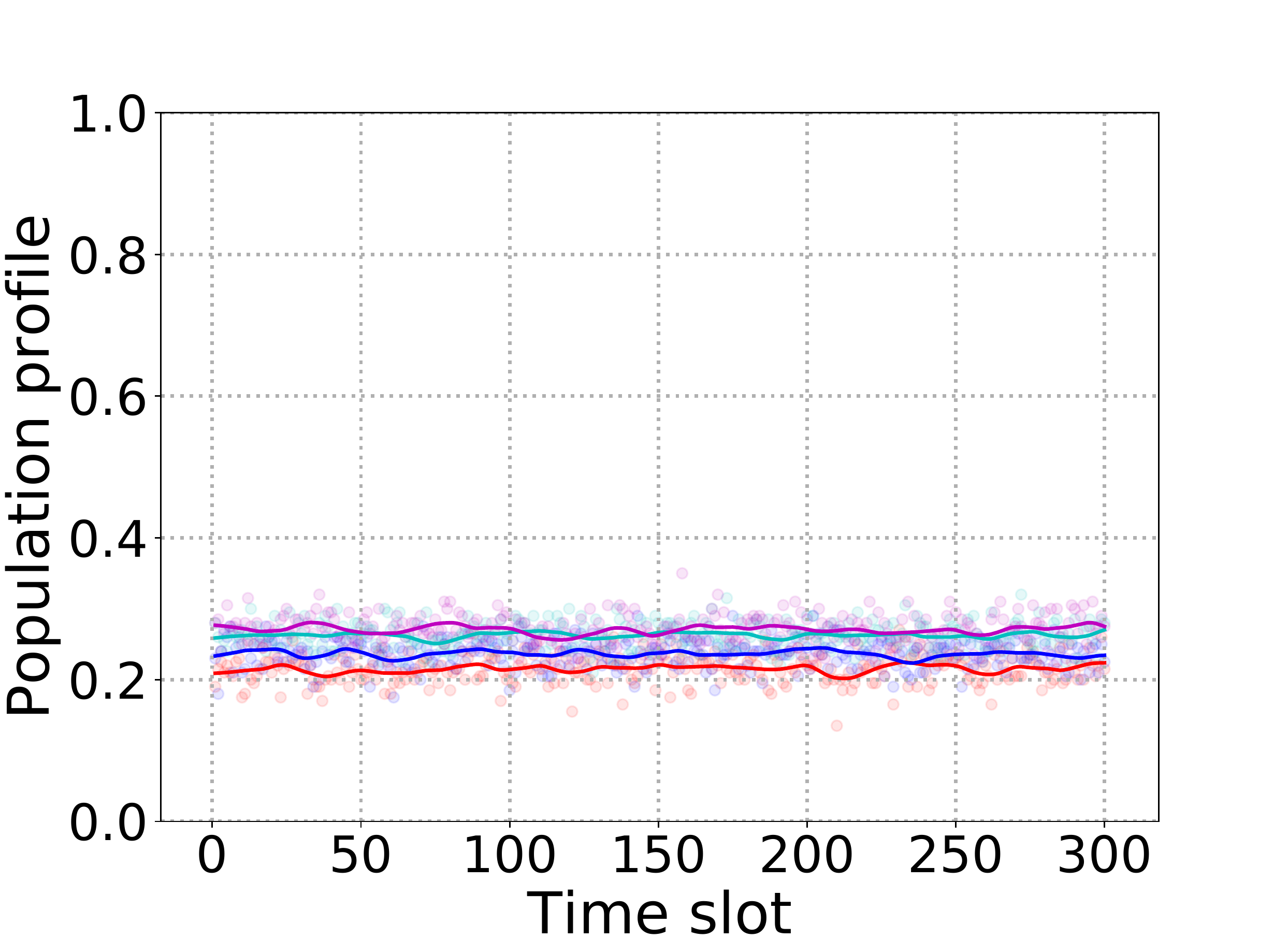}} \\
\vspace{-6pt}
\caption{Non-contraction mapping: population profile evolution.}
\label{Fig:GNConPop}
\vspace{-15pt}
\end{figure}

\emph{Contraction mapping}. Let $(\theta,\beta,\eta)$ be $(0.5, 0.5, 0.2)$, respectively, and hence the contraction condition $4 \theta (1-\eta) \beta < 1$ holds. Given the number of arms $M=4$, we run the bandit game for four times and display the state evolution of arm $2$  in Figure~\ref{Fig:GConState}. We can see that the state will converge to a steady value, which is also unique for different number of agents $N$, i.e., the bandit game has a unique MFE.

Figure~\ref{Fig:GConPop} shows the population profile of arm $2$, which  is unique and tends to be stable as $N$ increases. Because $\mathbb{E}\left[\mathbbm{1}_{\{\overline{a}^i = j\}}\right] = \sigma(\overline{s}^i,j)$ and $\mathrm{var}\left[\mathbbm{1}_{\{\overline{a}^i = j\}}\right] =  \sigma(\overline{s}^i,j)(1- \sigma(\overline{s}^i,j)) \le \frac{1}{4}$ due to $y(1-y) \le \frac{1}{4}, \forall y \in[0,1]$, we have $\mathrm{var}[f(\overline{\bm{s}}, j)] = \frac{1}{N^2} \sum_{i=1}^{N} \mathrm{var}\left[\mathbbm{1}_{\{\overline{a}^i = j\}}\right] \le \frac{1}{N^2} \frac{N}{4} = \frac{1}{4N}$.  Based on the Chebyshev's inequality, we obtain $\mathrm{Pr}\left(|f(\overline{\bm{s}}, j) - \mathbb{E}[f(\overline{\bm{s}}, j)]| \ge \epsilon \right) \le \frac{1}{4N\epsilon^2}$. Therefore, if $N$ increases, $f(\overline{\bm{s}}, j)$ will be more stable around $\mathbb{E}[f(\overline{\bm{s}}, j)]$. 

\emph{Non-contraction mapping}.  Let $\theta$, $\eta$ and $M$ stay the same, while $\beta$ changes to 30 so the contraction mapping condition $4 \theta (1-\eta) \beta < 1$ is violated. Similarly,  we run the bandit game for  four times and depict the state evolution in Figure~\ref{Fig:GNConState}, which shows that the state converges to multiple distinct MFEs. Moreover, we plot the population profile in Figure~\ref{Fig:GNConPop}. Comparing with Figure~\ref{Fig:GConPop}, the fluctuation around  $\mathbb{E}[f(\overline{\bm{s}}, j)]$ also becomes impaired when $N$ is large. Due to multiple MFEs, the population profile has various steady values as well.
\vspace{-3pt}
\subsection{Empirical Regret} \label{Sec:Eregret}
For the general reward, we compute the regret when contraction mapping condition holds, i.e., $(\theta,\beta,\eta)=(0.5,0.5,0.2)$. Furthermore, we  implement a linear reward: $r(f_n(j), j) = 1 - \theta(j) f_n(j)$ where $\theta(j)  \in [0.8\theta, \theta], \forall j \in \mathcal{M}$ and $(\theta,\beta, \eta)=(1,2,0.2)$, so the reward  is a contraction from  Corollary~\ref{Cor:LinearContr}.  We run the evaluation for six times with each operating for $T=2000$ time slots, and show the average regret and cumulative rewards  in Table~\ref{Tab:regret}. A well-known regret bound  for EXP3 is $O(\sqrt{T})$, and here $\sqrt{T} = 44.721$. We can see that both regrets are much smaller than  $\sqrt{T}$ for each $N$. Besides, the regret of   general reward is  less than that of linear reward, as changes in the population profile have smaller impact when appearing in the denominator of Eq.~\eqref{Eq:general}.  Moreover, the regrets decease as $N$ grows large for both cases, which is mainly due to a more stable population profile. In summary, the stationary policy has a tight empirical regret.
\vspace{-5pt}
\begin{table}[H]
\centering
\begin{tabular}{llccc}
\toprule
Reward & Term & $N=50$ & $N=100$  & $N=200$ \\
\midrule
General & Regret  & 14.653 & 13.758 & 7.787\\
	     &Rewards & 1791.904 & 1796.961 & 1791.061\\
Linear   &Regret & 21.464 & 19.023 & 17.932\\
	    &Rewards & 1541.083 & 1554.948 & 1560.613\\
\bottomrule
\end{tabular}
\vspace{-5pt}
\caption{Empirical regret}\label{Tab:regret}
\vspace{-15pt}
\end{table}

\section{Conclusion}\label{Sec:conclusion}
We propose a mean field model  to study a  large-population bandit game with  a continuous reward. Concretely, we  characterize the existence and uniqueness of MFE by showing the state evolution is  upper semi-continuous, and deriving  contraction mapping conditions based on the stochastic approximation, respectively. Extensive evaluations are performed to validate our  mean field analysis and tight empirical regret.
\section*{Acknowledgments}
This work was supported in part by NSF China under Grant 61902358, Zhejiang Provincial Natural Science Foundation of China under Grant LQ19F020007.

\bibliographystyle{named}
\bibliography{ijcai21}

\newpage

\begin{appendices}
\appendix
\section{Proof of Lemma~\ref{Lem:pseudo}} \label{App:pseudo}
\subsection{Preliminary}
\subsubsection{Asymptotic pseudotrajectory}
We first introduce the concept of the asymptotic pseudotrajectory~\cite{Ref:dynamics1999Benaim}.  Suppose a continuous mapping $\Phi$ on the space $\mathbb{R}^{N \times M}$ is a semiflow:
\begin{displaymath}
\begin{aligned}
&\Phi :  \mathbb{R}_+ \times \mathbb{R}^{N \times M} \rightarrow \mathbb{R}^{N \times M},\\
 &(t,\bm{s}) \rightarrow \Phi(t,\bm{s}) : = \Phi_t(\bm{s}),
\end{aligned}
\end{displaymath}
such that
\begin{displaymath}
\begin{aligned}
\Phi_0 = \mathrm{Identity},~~\Phi_{t+h} = \Phi_t  \circ \Phi_h.
\end{aligned}
\end{displaymath}
In fact, $\Phi_t(\bm{s})$ can be interpreted as  the state evolution at the continuous timescale.  Since $\bm{s}_t=\left[s^1_t, s^2_t,...,s^N_t\right]$ denotes the state profile at time $t$, we can also regard it as a mapping from $t \in \mathbb{R}_+$ to  the set   $\mathbb{R}^{N \times M}$.

\begin{dfn} \label{Dfn:pseudo}
A continuous mapping $\bm{s}: \mathbb{R}_+ \rightarrow \mathbb{R}^{N \times M}$ is an asymptotic pseudotrajectory for $\Phi$ if:
\begin{equation}
\lim_{t \rightarrow \infty}  \sup_{0 \le h \le T} d \left(\bm{s}_{t+h}, \Phi_h(\bm{s}_t)\right)=0, \forall T>0,
\end{equation}
where $d(\cdot)$ is a distance measure.
\end{dfn}
From this definition, if $\bm{s}_{t+h}$ is an  asymptotic pseudotrajectory for $\Phi_h(\bm{s}_t)$, then they  have the \emph{same convergence property}.  To prove  Lemma~\ref{Lem:pseudo}, we will derive a continuous-time ODE of  $\bm{s}_t$, for which the interpolated process of the discrete-time states $\bm{s}_n$ is an asymptotic pseudotrajectory.

\subsubsection{Interpolated process}
 Let $\bm{w}_n = [w_n^1,w_n^2,...,w_n^N]$, and define a filtration $\{\mathcal{F}_n\}_{n \ge 0}$ generated by stochastic processes $\{\bm{s}_n, \bm{w}_n\}_{n \ge 0}$. Obviously, $\mathcal{F}_n$ is a $\sigma$-algebra with $\mathcal{F}_n \subset \mathcal{F}_{n+1}$.  As $u^i_{n}(j) = w^i_n(j) - \mathbb{E}[w^i_n(j)]$, we know that the process $\{u^i_n\}_{n \ge 0}$ is a martingale, and $\mathbb{E}[u^i_{n}(j)| \mathcal{F}_{n}]=0, \forall j \in \mathcal{M}$.  
 
The state profile $\bm{s}_n$ is at the discrete timescale, while the asymptotic pseudotrajectory involves two continuous-time processes. Therefore, we need to introduce a continuous-time  interpolated  process of $\bm{s}_n$ to help link $\bm{s}_n$ to $\bm{s}_t$.  Let $\tau_0 =0$ and $\tau_n = \sum_{k=0}^{n-1} \gamma_k, \forall n \ge 1$. Define the interpolated process of  state profile $\bm{s}_n$ as $\tilde{s}^i_{\tau_n + h} = s^i_n + h \frac{s^i_{n+1} - s^i_n}{\tau_{n+1} -\tau_n}, 0<h<\gamma_{n}, \forall i \in \mathcal{N}$. Intuitively, the state $s_n^i$ and its interpolated process $\tilde{s}^i_t$ have the same trajectories. We can analyze the convergence of $\bm{s}_n$ by characterizing a deterministic process $\bm{s}_t$, for which  the \emph{interpolated process $\tilde{\bm{s}}_t$ is an asymptotic pseudotrajectory}. The essence of Lemma~\ref{Lem:pseudo} is to show that the interpolated process $\tilde{s}^i_t$ is indeed an asymptotic pseudotrajectory for the solution $s^i_t$ to the ODE Eq.~\eqref{Eq:ode}. In the following, we present the detailed proof of Lemma~\ref{Lem:pseudo}.

\subsection{Detailed proof}\label{App:detailed_proof}
We first demonstrate $\lim_{n \rightarrow \infty} \gamma_n =0$. Suppose that $\gamma_n \nrightarrow 0$, and then there exists a  value $\epsilon >0$ such that $\gamma_n \ge \epsilon, \forall n$. Hence, $\sum_{n=0} ^ {K-1} \gamma_n^2 \ge K \epsilon^2 \rightarrow \infty$ as $K \rightarrow \infty$, which is contradictory to Eq.~\eqref{Eq:timestep}. Let $m(t) = \sup \{n \ge 0, t \ge \tau_n\}$.  According to  Robbins-Monro theorem~\cite{Ref:dynamics1999Benaim},  to prove $\tilde{s}^i_t$ is the asymptotic pseudotrajectory for the ODE Eq.~\eqref{Eq:ode}, we need to show that the discrete-time processes $s^i_n, u^i_n$ satisfy the following two conditions.

1) For all $T >0$, $\lim_{n \rightarrow \infty} \sup \{||\sum_{l=n}^{k-1} \gamma_{l}u^i_{l}||_2^2: k= n+1, ..., m(\tau_n +T)\}=0$. From the H\"older's inequality, we have:
\begin{displaymath}
||\sum_{l=n}^{k-1} \gamma_{l}u^i_{l} ||_2^2 \le \sum_{l=n}^{k-1} \gamma_{l}^2||u^i_{l}||_2^2.
\end{displaymath}
Since $\lim_{n \rightarrow \infty} \{\sum_{l=n}^{k-1} \gamma_{l}^2: k= n+1, ..., m(\tau_n +T) \}=0$,  we are left to demonstrate $||u^i_{l}||_2^2$ is  finite. In fact, $u^i_{l} = w^i_l - \mathbb{E}[w^i_l]$, and then $||u^i_{l}||_2^2 = ||w^i_l - \mathbb{E}[w^i_l]||_2^2 \le ||w^i_l||_2^2 + ||\mathbb{E}[w^i_l]||_2^2 \le 2M$. As a result,  we claim condition 1) holds.

2) $\sup_{n} ||s^i_n||_2^2$ is bounded. This is naturally satisfied due to $s^i_n(j) \in [0,1], \forall j \in \mathcal{M}$.

Overall, the state $s_n^i$ will converge to $s_t^i$ specified by Eq.~\eqref{Eq:ode} when $n$ and $t$ go to infinity.

\section{Proof of Theorem~\ref{Thm:cumulative}} \label{App:cumulative}
Let $X^i_n(j) = \mathrm{Exp}(\beta s^i_n(j))$ and $X^i_n = \sum_{j=1}^M X^i_n(j)$. Combining with Eqs.~\eqref{Eq:hedge} and~\eqref{Eq:StateUpdate}, we obtain:
\begin{equation}
\begin{aligned}
\frac{X^i_{n+1}}{X^i_n} &= \sum_{j=1}^{M} \frac{X_{n+1}(j)}{X_n}\\
&= \sum_{j=1}^M \frac{X^i_n(j)}{X^i_n} \mathrm{Exp}[\beta \gamma_{n} (w^i_n(j)-s^i_n(j))]\\
&= \sum_{j=1}^M \frac{\sigma(s^i_n,j)-\frac{\eta}{M}}{1-\eta} \mathrm{Exp}[\beta \gamma_{n} (w^i_n(j)-s^i_n(j))]\\
& \le \sum_{j=1}^M \frac{\sigma(s^i_n,j)-\frac{\eta}{M}}{1-\eta} [1+ \beta \gamma_{n} (w^i_n(j)-s^i_n(j))\\
& + (e-2)\beta^2 \gamma^2_{n} (w^i_n(j)-s^i_n(j))^2]\\
& \le 1 + \sum_{j=1}^M  \frac{\sigma(s^i_n,j)-\frac{\eta}{M}}{1-\eta} \beta \gamma_{n}(w^i_n(j)-s^i_n(j))\\
& + \sum_{j=1}^M  \frac{(e-2)\sigma(s^i_n,j)\beta^2 \gamma^2_{n}}{1-\eta} (w^i_n(j)-s^i_n(j))^2,
\end{aligned}
\end{equation}
where the first inequality is because  $\mathrm{Exp}(y) \le 1+y+(e-2)y^2$ and $e$ is the Euler's number. Using the result $\ln y \le y-1, \forall y>0$ and the fact that $\ln \frac{X^i_{K+1}}{X^i_0} = \sum_{n=0}^{K} \ln \frac{X^i_{n+1}}{X^i_n}$, we have:
\begin{equation} \label{Eq:change1}
\begin{aligned}
&\ln \frac{X^i_{K+1}}{X^i_0}  \le  \sum_{n=0} ^{K} \sum_{j=1}^M  \frac{\sigma(s^i_n,j)-\frac{\eta}{M}}{1-\eta} \beta \gamma_{n}(w^i_n(j)-s^i_n(j))\\
&+ \sum_{n=0} ^{K} \sum_{j=1}^M  \frac{(e-2)\sigma(s^i_n,j)\beta^2 \gamma^2_{n}}{1-\eta} (w^i_n(j)-s^i_n(j))^2.
\end{aligned}
\end{equation}

For any arm $j$, it satisfies:
\begin{equation} \label{Eq:change2}
\ln \frac{X^i_{K+1}}{X^i_0} \ge \ln \frac{X^i_{K+1}(j)}{X^i_0}.
\end{equation} 
Comparing Eqs.~\eqref{Eq:change1} and~\eqref{Eq:change2}, we obtain:  
\begin{equation}
\begin{aligned}
&\ln \frac{X^i_{K+1}(j)}{X^i_0} \le  \sum_{n=0} ^{K} \sum_{j=1}^M  \frac{\sigma(s^i_n,j)-\frac{\eta}{M}}{1-\eta} \beta \gamma_{n}(w^i_n(j)-s^i_n(j))\\
&+ \sum_{n=0} ^{K} \sum_{j=1}^M  \frac{(e-2)\sigma(s^i_n,j)\beta^2 \gamma^2_{n}}{1-\eta} (w^i_n(j)-s^i_n(j))^2.
\end{aligned}
\end{equation}
Moreover, $\ln \frac{X^i_{K+1}(j)}{X^i_0}  = \beta s^i_0(j) + \sum_{n=0}^{K} \beta \Delta s^i_n(j) -\ln (\sum_{j=1}^{M} \mathrm{Exp}(\beta s^i_0(j)))$. Using the expression of the inner product between two vectors, we complete the proof.

\section{Additional Evaluation Results} \label{App:linear}
\subsection{Results for Linear Reward}
\emph{Reward function.} In line with the evaluated resource competition game, we consider the linear reward $r(f_n(j),j)$ has the following expression:
\begin{equation} \label{Eq:linear}
r(f_n(j), j) = 1 - \theta(j) f_n(j),
\end{equation}
where  $\theta(j) \in [0.8\theta,\theta], \forall j \in \mathcal{M}$ with $\theta \in [0,1]$. It can be verified that $r(f_n(j),j)$ falls into the range $[0,1]$ and maintains $\theta$-Lipschitz continuity as well.

\emph{Contraction linear reward}. Similarly, we illustrate the results when $||\cdot||_{\infty}$-contraction mapping stands in the first place. According to Corollary~\ref{Cor:LinearContr}, values of $(\theta,\beta,\eta)$ are assigned to $(1,2,0.2)$, respectively, thereby meeting the contraction condition $\frac{\theta (1-\eta) \beta}{2} < 1$. The arm number $M=4$, and each state is initialized to be a value in $[0,1]$. Still run the bandit game for four times, and the state evolution of a selected arm is depicted in Figure~\ref{Fig:LConState}. We can see that the curves will approach a specific value for different agent number $N$, i.e., unique MFE is derived due to contraction mapping of the reward function. As for the population profile shown in Figure~\ref{Fig:LConPop}, we also obtain that it tends to be more stable around a unique value when $N$ increases, which can also be explained by the Chebyshev's inequality as aforementioned.

\begin{figure}[t]
  \centering
  \subfloat[$N=50$]{\includegraphics[width=0.325\linewidth,height=0.67in]{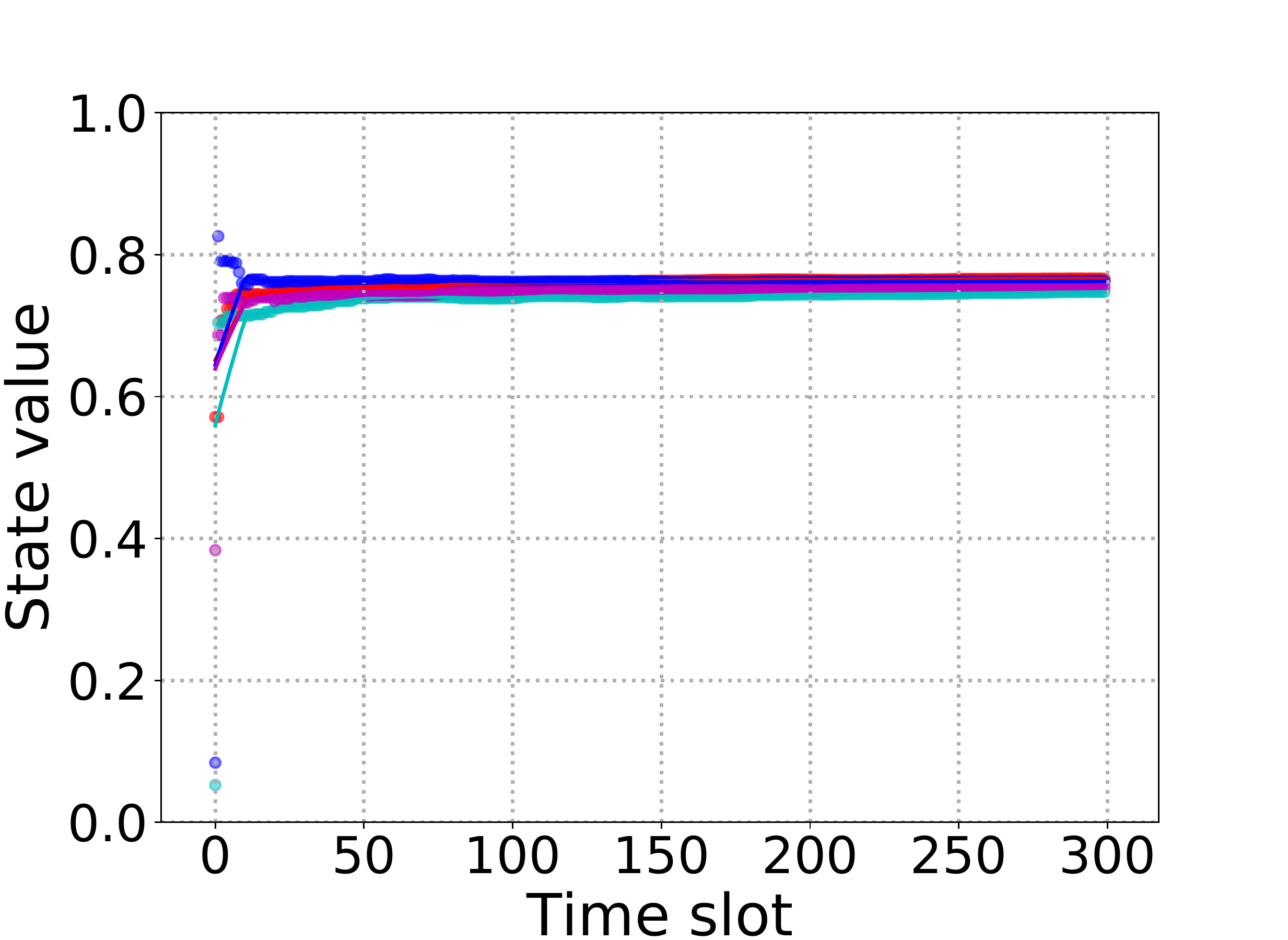}} \ 
  \subfloat[$N=100$]{\includegraphics[width=0.325\linewidth,height=0.67in]{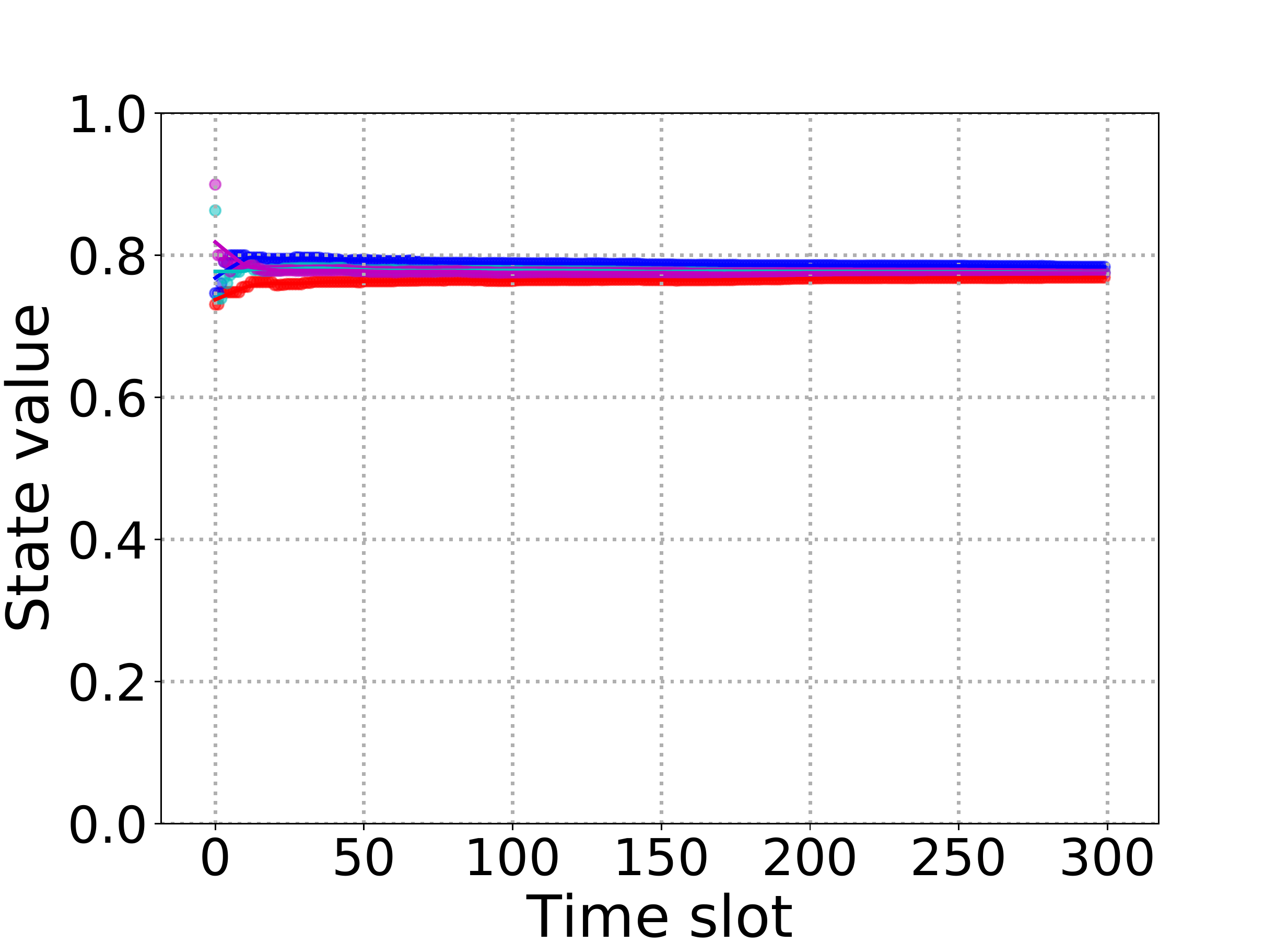}} \
  \subfloat[$N=200$]{\includegraphics[width=0.325\linewidth,height=0.67in]{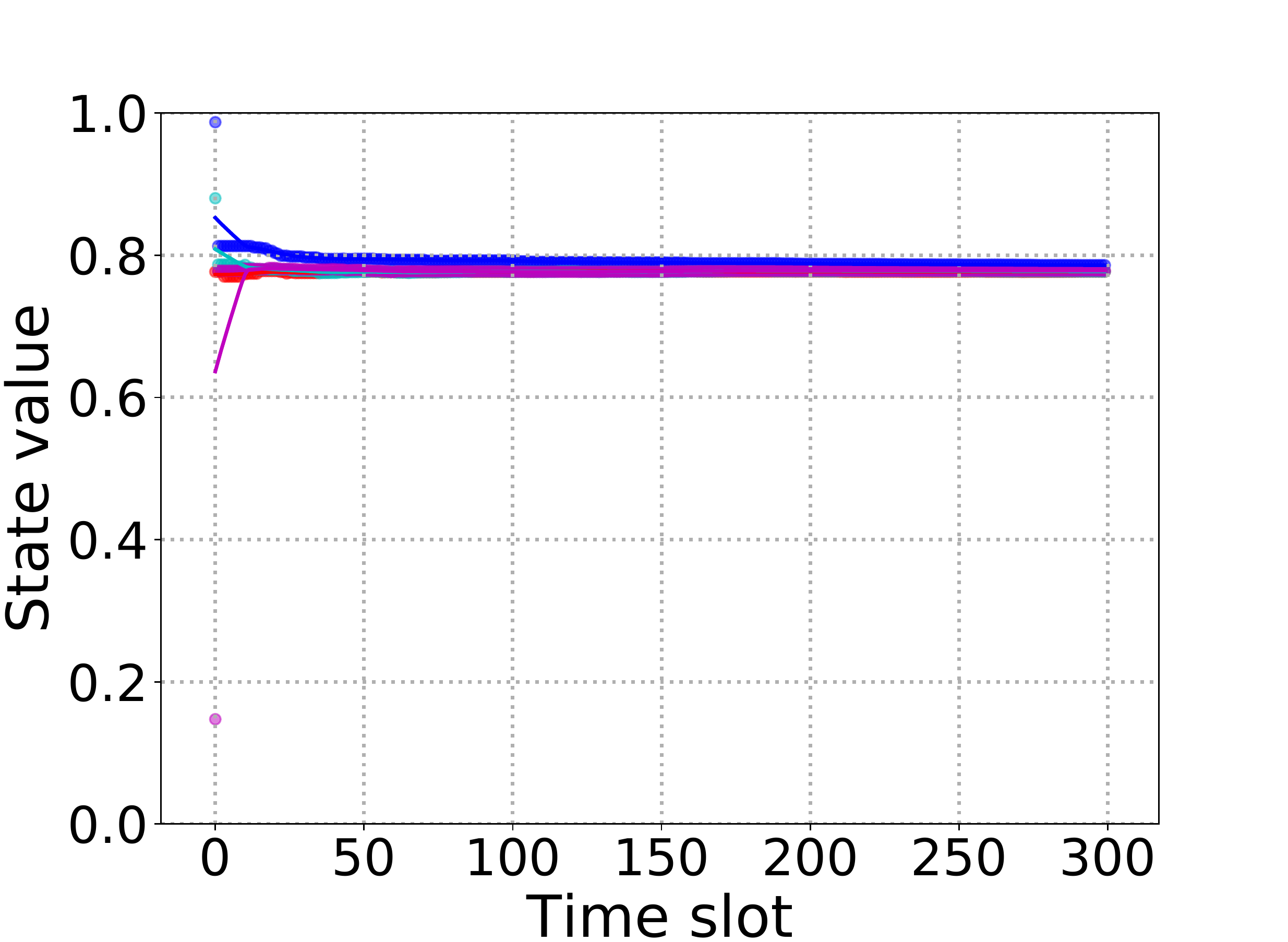}} \\
   \vspace{-6pt}
  \caption{Contraction linear reward: state evolution.}
  \label{Fig:LConState}
  \vspace{-15pt}
\end{figure}

\begin{figure}[t]
  \centering
  \subfloat[$N=50$]{\includegraphics[width=0.325\linewidth,height=0.67in]{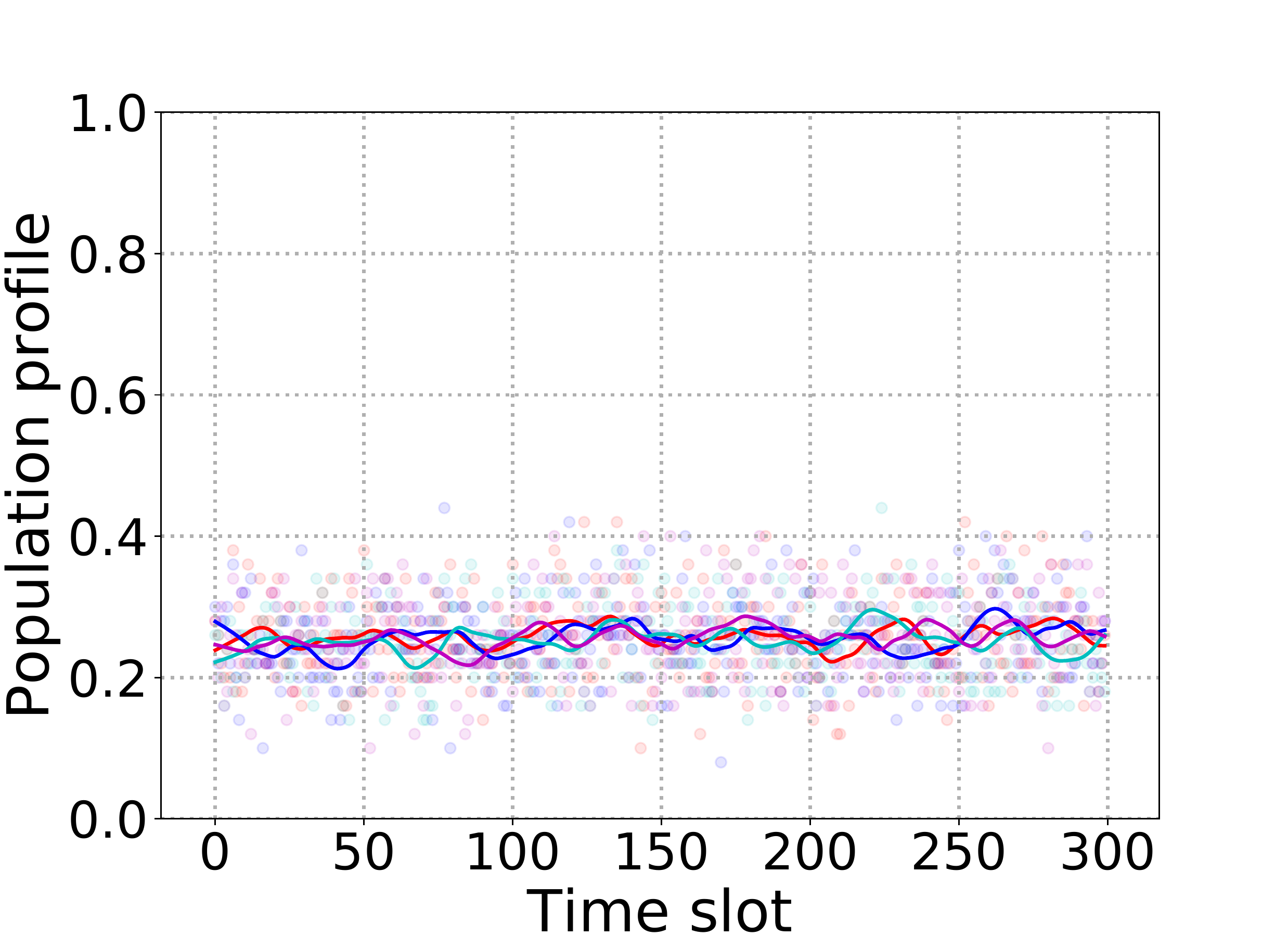}} \ 
  \subfloat[$N=100$]{\includegraphics[width=0.325\linewidth,height=0.67in]{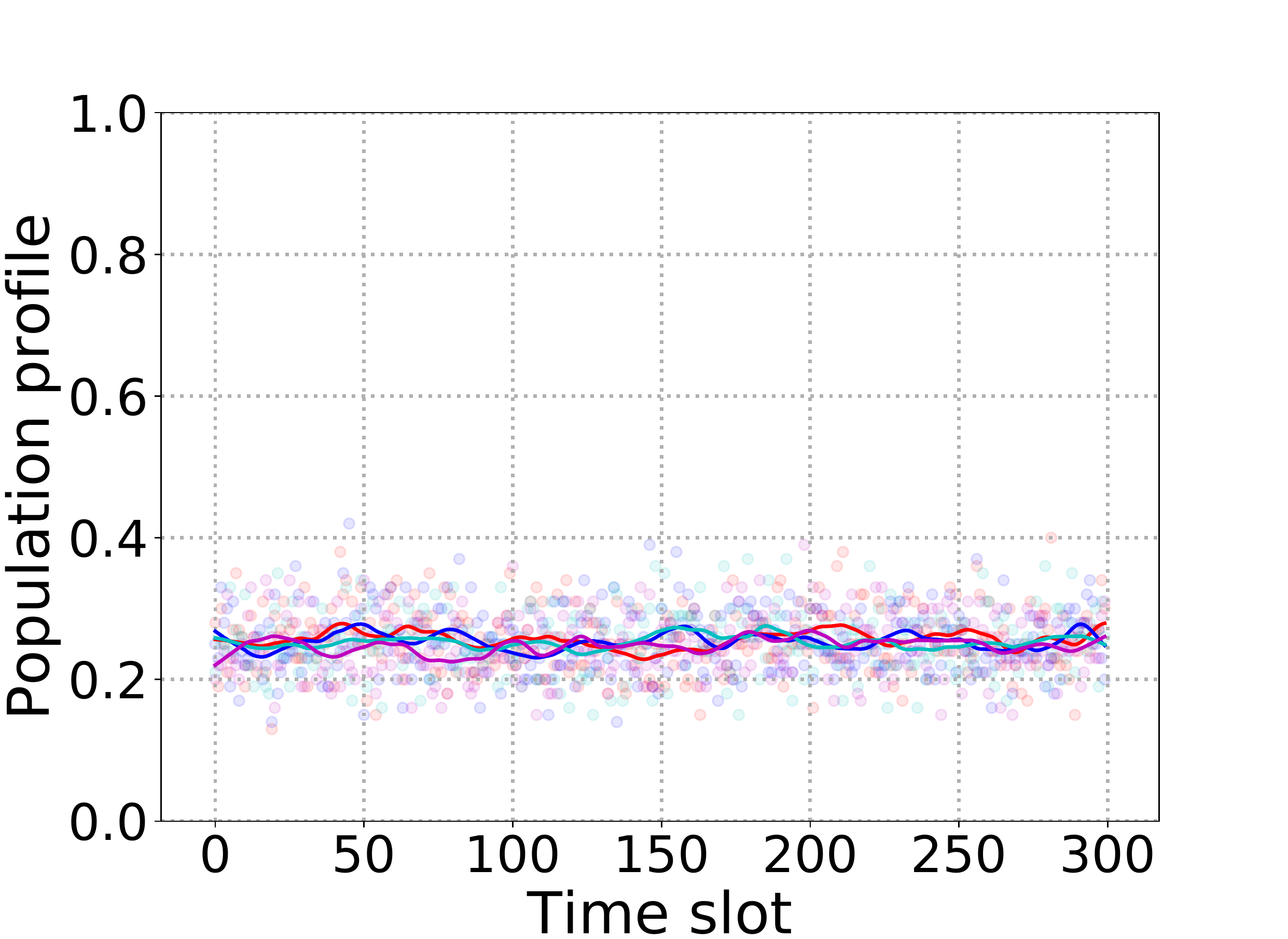}} \
  \subfloat[$N=200$]{\includegraphics[width=0.325\linewidth,height=0.67in]{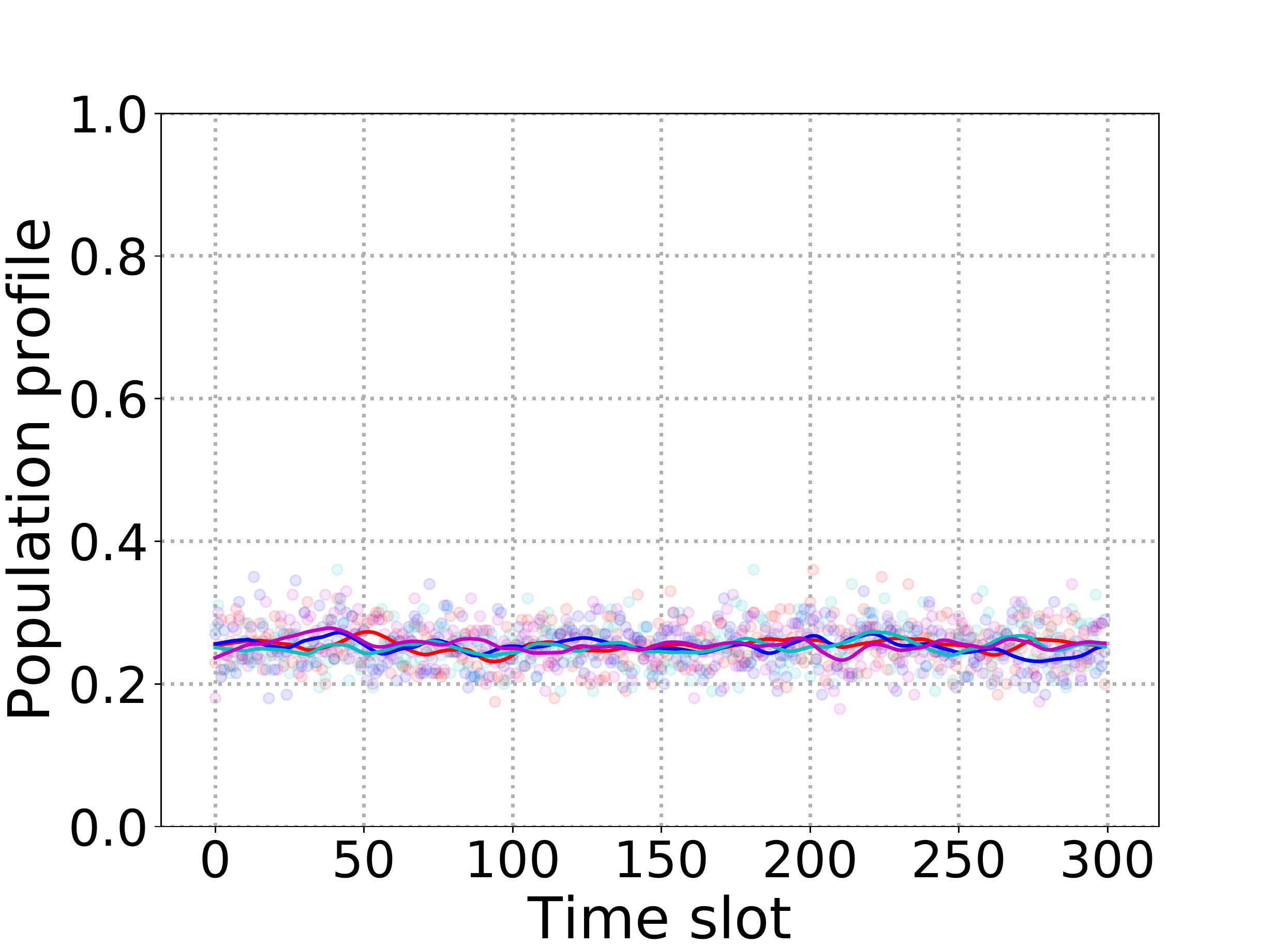}} \\
  \vspace{-6pt}
  \caption{Contraction linear reward: population profile evolution.}
  \label{Fig:LConPop}
   \vspace{-15pt}
\end{figure}

\begin{figure}[t]
  \centering
  \subfloat[$N=50$]{\includegraphics[width=0.325\linewidth,height=0.67in]{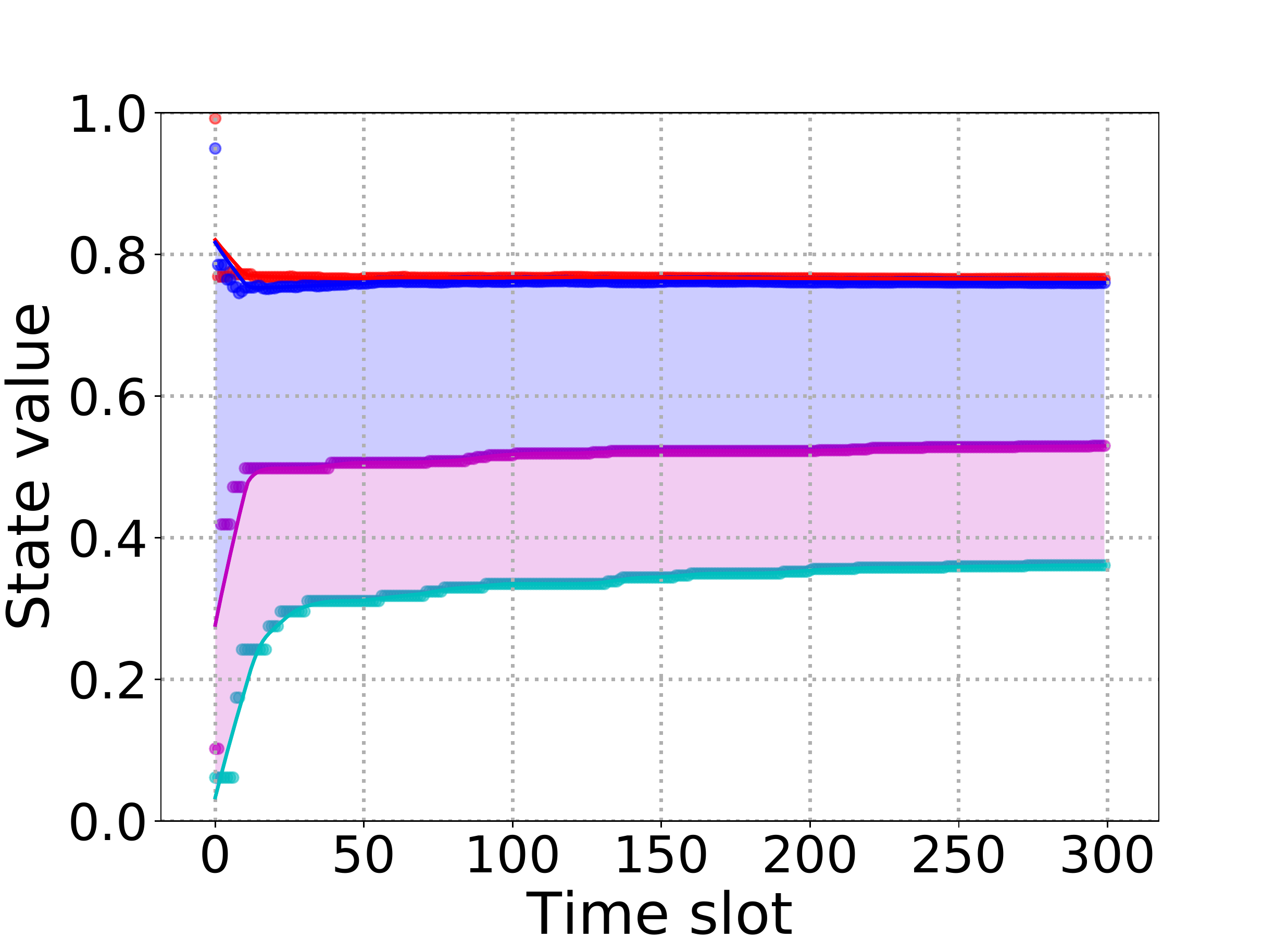}} \ 
  \subfloat[$N=100$]{\includegraphics[width=0.325\linewidth,height=0.67in]{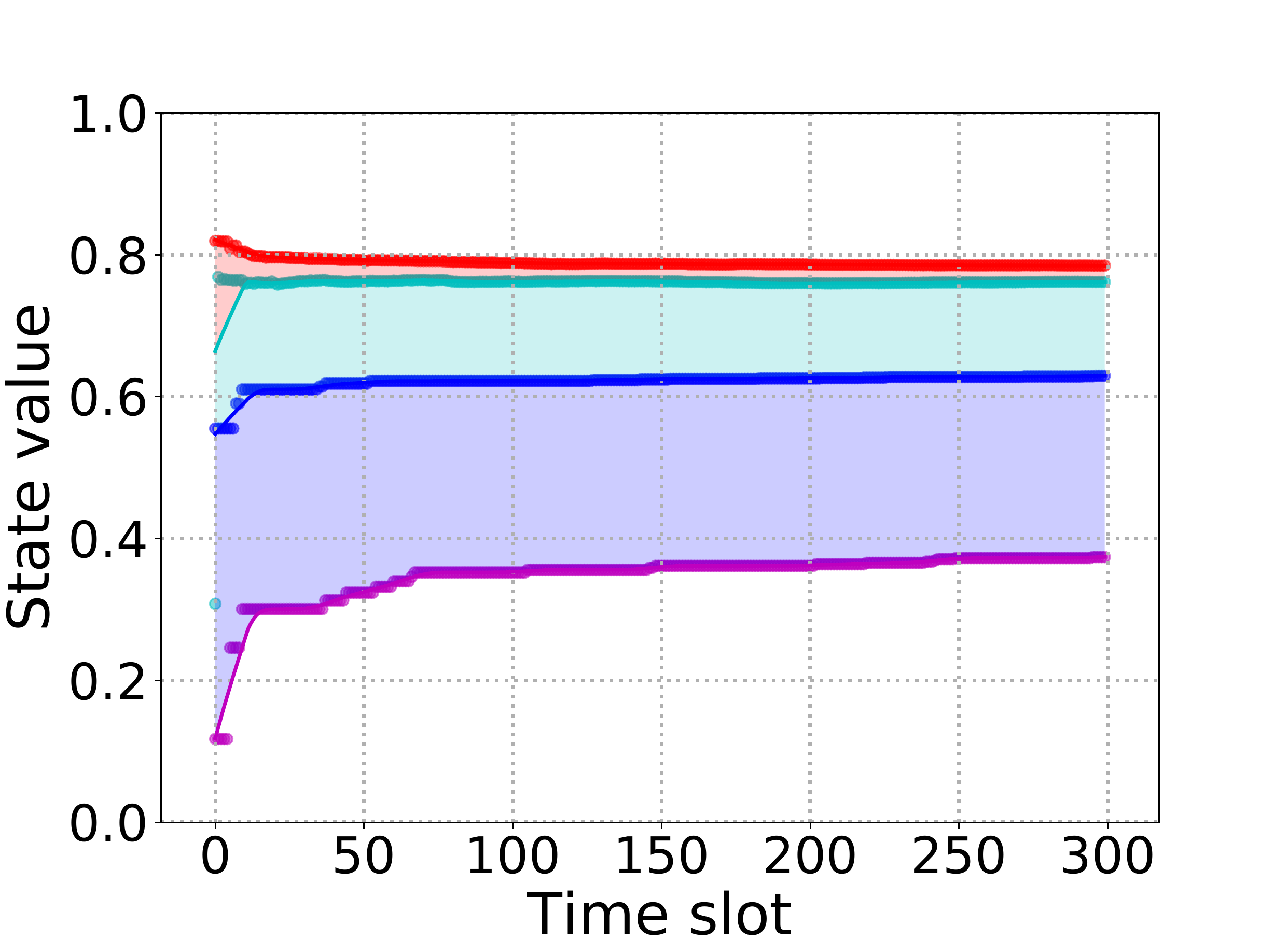}} \
  \subfloat[$N=200$]{\includegraphics[width=0.325\linewidth,height=0.67in]{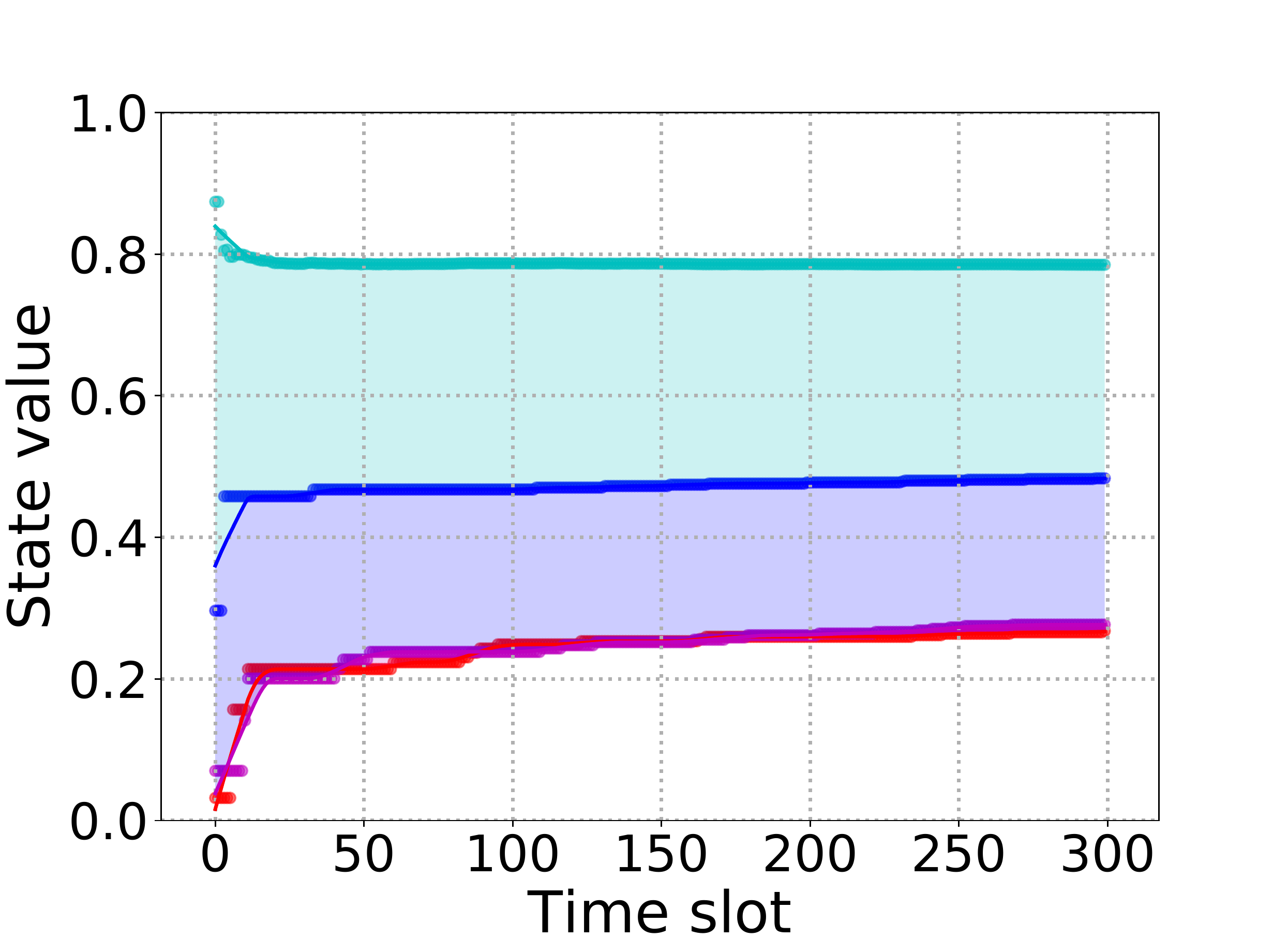}} \\
  \vspace{-6pt}
  \caption{Non-contraction linear reward: state evolution.}
  \label{Fig:LNConState}
   \vspace{-15pt}
\end{figure}

\begin{figure}[t]
  \centering
  \subfloat[$N=50$]{\includegraphics[width=0.325\linewidth,height=0.67in]{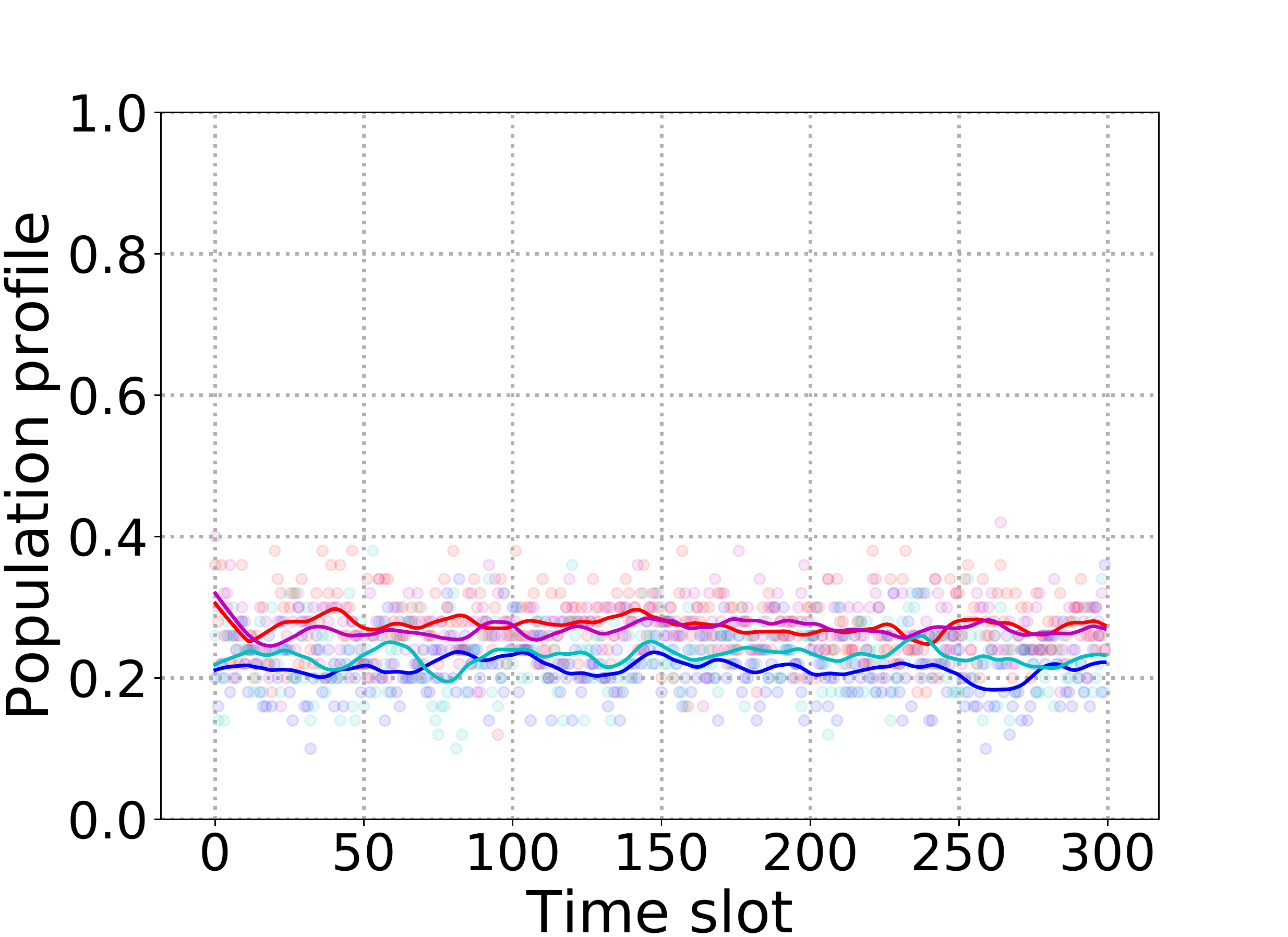}} \ 
  \subfloat[$N=100$]{\includegraphics[width=0.325\linewidth,height=0.67in]{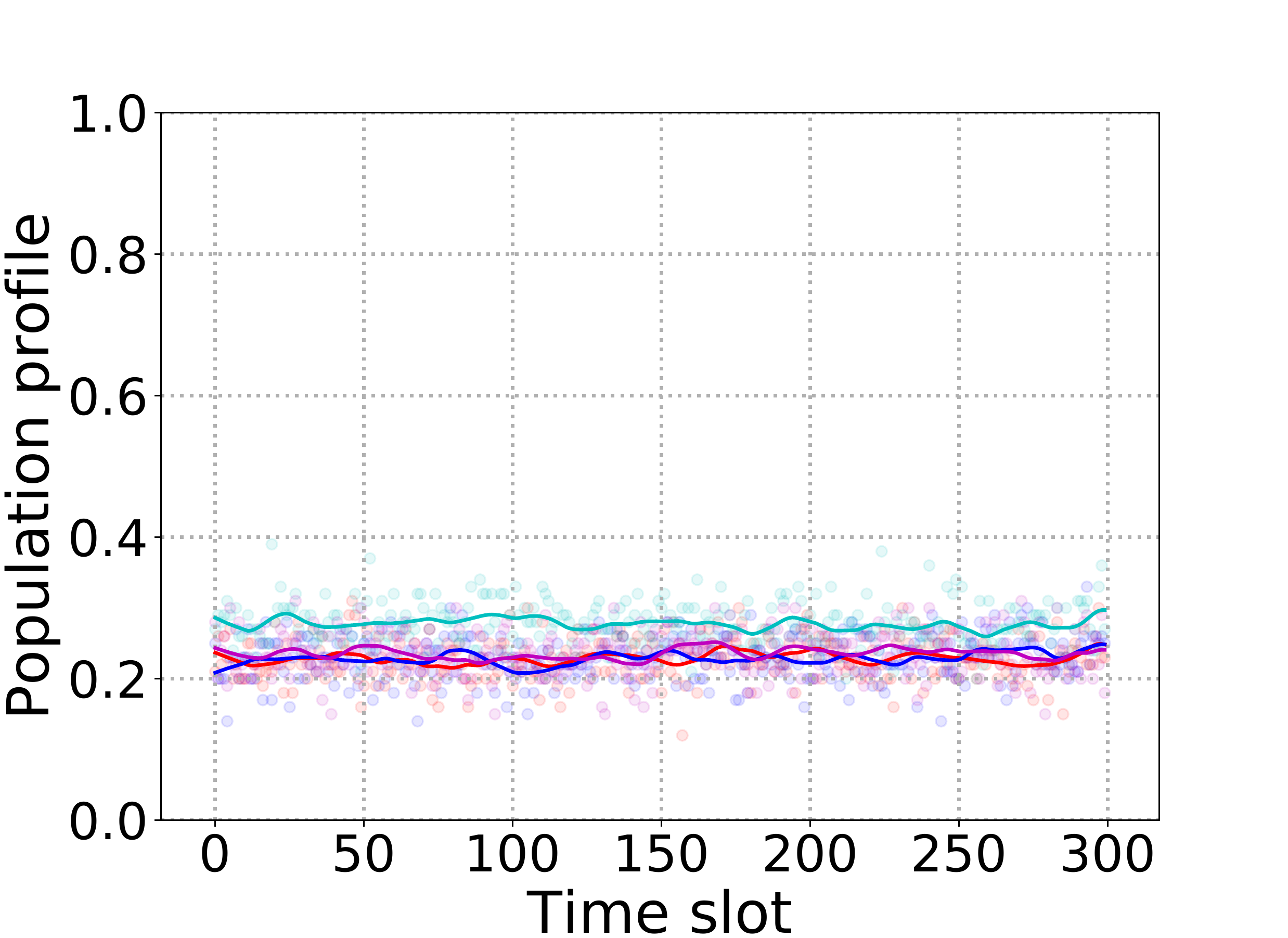}} \
  \subfloat[$N=200$]{\includegraphics[width=0.325\linewidth,height=0.67in]{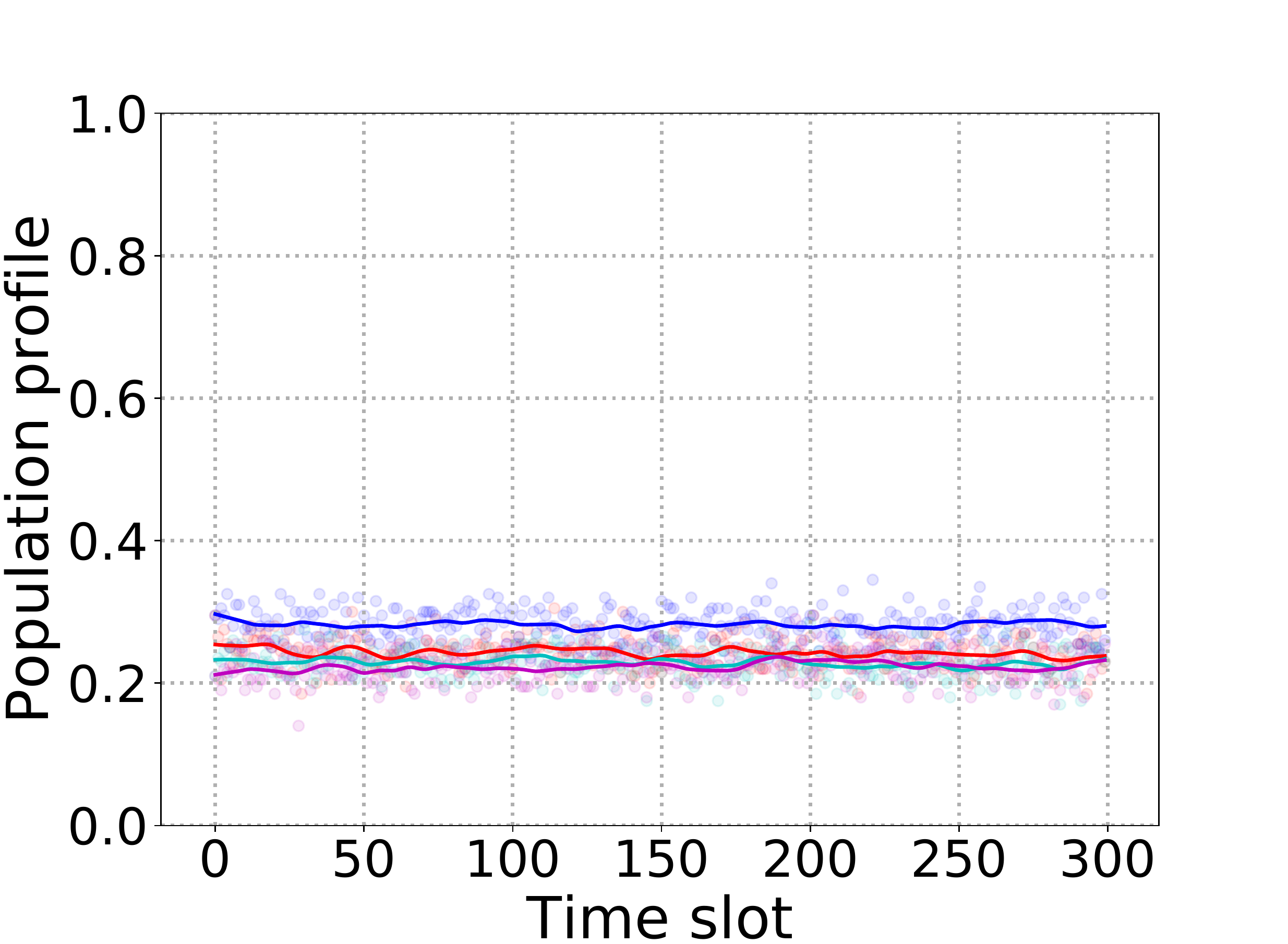}} \\
  \vspace{-6pt}
  \caption{Non-contraction linear reward: population profile evolution.}
  \label{Fig:LNConPop}
   \vspace{-15pt}
\end{figure}

\emph{Non-contraction linear reward}. Now we discuss about the situation where the contraction mapping condition is no longer satisfied. In particular, parameters $\theta$, $\eta$ and $M$ remain unchanged, while $\beta$ is altered to 40. The bandit game runs for four times with state being initialized each time. We show the results for the evolution of  state and population profile in Figures~\ref{Fig:LNConState} and~\ref{Fig:LNConPop}, respectively. Like the general reward, the state will reach different fixed points, which amount to multiple MFEs. Moreover, the population profile also converges to various steady values with the fluctuation being impaired for larger agent number $N$.

\subsection{Empirical Regret under Non-contraction  Mapping}
Aside from the empirical regret demonstrated in Section~\ref{Sec:Eregret},  we further provide the results if the general/linear reward function does not satisfy the contraction mapping condition, i.e.,  $(\theta,\beta,\eta)=(0.5,30,0.2)$ for general reward and  $(\theta,\beta,\eta)=(1,40,0.2)$ for linear reward. Compared with the contraction case, agent states may diverge to various MFEs. Still run the evaluation for six times,  where each run lasts for $T=2000$ time slots in total. Both the cumulative rewards and regrets for the two reward functions are exhibited in Table~\ref{Tab:regret1}. For the cumulative rewards, both general function and linear function yield similar results to those for the contraction case. Nevertheless, the regrets have different behaviors. Pertaining to non-contraction reward functions, the corresponding regrets have large variances, with the average values bigger than those for the contraction case. One of the main reasons behind is that there exist multiple MFEs, which will cause varied and large (on average) regrets. But at the same time, the regrets for general and linear functions are still less than the bound $\sqrt{T} = 44.721$. Besides, regrets tend to decrease as  agent number $N$ increases large, and the general function  has smaller regret than the linear function as well. To sum up,  the stationary policy  still has tight empirical regrets for non-contraction reward functions.
\vspace{-8pt}
\begin{table}
\centering
\begin{tabular}{llccc}
\toprule
Reward & Term & $N=50$ & $N=100$  & $N=200$ \\
\midrule
General & Regret  & 21.187 & 19.014 & 12.108\\
	     & Rewards & 1786.238 & 1784.582 & 1791.829\\
Linear   &Regret & 28.362 & 26.626 & 23.947\\
	    &Rewards & 1552.282 & 1558.675 & 1559.661\\
\bottomrule
\end{tabular}
\caption{Empirical regret under non-contraction reward}\label{Tab:regret1}
\end{table}

\end{appendices}

\end{document}